\documentclass[a4paper,english]{lipics-v2018}

\usepackage{microtype}

\usepackage{stmaryrd}
\usepackage{tikz}
\usepackage{wrapfig}
\usetikzlibrary{math}
\usetikzlibrary{arrows.meta}
\usetikzlibrary{positioning}

\usepackage{color}
\usetikzlibrary{matrix,arrows}

\title{Quantitative bisimulations using coreflections and open morphisms}

\titlerunning{Quantitative bisimulations using coreflections and open morphisms}

\author{J\'er\'emy Dubut}{National Institute of Informatics, Tokyo, Japan\\Japanese-French Laboratory for Informatics, Tokyo, Japan}{}{}{}

\author{Ichiro Hasuo}{National Institute of Informatics and SOKENDAI, Tokyo, Japan}{}{}{}

\author{Shin-ya Katsumata}{National Institute of Informatics, Tokyo, Japan}{}{}{}

\author{David Sprunger}{National Institute of Informatics, Tokyo, Japan}{}{}{}

\authorrunning{J. Dubut, I. Hasuo, S. Katsumata and D. Sprunger}

\Copyright{J\'er\'emy Dubut, Ichiro Hasuo, Shin-ya Katsumata and David Sprunger}

\subjclass{Theory of computation $\rightarrow$ Models of computation $\rightarrow$ Probabilistic computation, Theory of computation $\rightarrow$ Models of computation $\rightarrow$ Timed and hybrid models, Theory of computation $\rightarrow$ Semantics and reasoning}


\category{}

\relatedversion{}

\supplement{}

\funding{The authors were supported by ERATO HASUO Metamathematics for Systems Design Project (No. JPMJER1603), JST.}

\EventEditors{John Q. Open and Joan R. Access}
\EventNoEds{2}
\EventLongTitle{42nd Conference on Very Important Topics (CVIT 2016)}
\EventShortTitle{CVIT 2016}
\EventAcronym{CVIT}
\EventYear{2016}
\EventDate{December 24--27, 2016}
\EventLocation{Little Whinging, United Kingdom}
\EventLogo{}
\SeriesVolume{42}
\ArticleNo{23}
\nolinenumbers 
\hideLIPIcs  

\begin{document}

\theoremstyle{plain}
\newtheorem{proposition}{Proposition} 

\newcommand\map[3]{#1 : #2 \longrightarrow #3}
\newcommand\mapa[3]{#2 \stackrel{#1}{\longrightarrow} #3}
\newcommand\maps[2]{#1 \longrightarrow #2}
\newcommand\natr[3]{#1 : #2 \Longrightarrow #3}
\newcommand\reach{\textbf{Reach}_{\Sigma}}

\newcommand\ts{\textbf{TS}_{\Sigma}}
\newcommand\tso{\textbf{TS}_{\Sigma}\OO}

\newcommand\lin{\textbf{Lin}_{\Sigma}}
\newcommand\lino{\textbf{Lin}_{\Sigma}\OO}

\newcommand\tr{\textbf{Tr}_{\Sigma}}
\newcommand\tro{\textbf{Tr}_{\Sigma}\OO}

\newcommand\OO{\mathcal{O}}
\newcommand\PP{\mathcal{P}}
\newcommand\M{\mathcal{M}}
\newcommand\C{\mathcal{C}}
\newcommand\D{\mathcal{D}}
\newcommand\trcm{\textbf{Tr(}\M\textbf{,}\PP\textbf{)}}
\newcommand\prob{\textbf{Prob}_{\Sigma}}
\newcommand\rea{\mathbb{R}_{\geq 0}}
\newcommand\sem[2]{\llbracket #1 \rrbracket_{#2}}
\newcommand\tts{\textbf{TTS}_{\Sigma}}
\newcommand\tsrea{\textbf{TS}_{\Sigma\times\mathbb{R}_{>0}}}
\newcommand\trrea{\textbf{Tr}_{\Sigma\times\mathbb{R}_{>0}}}
\newcommand\tsorea{\textbf{TS}_{\Sigma\times\rea}\OO}
\newcommand\trorea{\textbf{Tr}_{\Sigma\times\rea}\OO}
\newcommand\hysy{\textbf{HS}_{\Sigma}}
\newcommand\unf{\text{unf}}
\newcommand{\Ext}{\mathrm{Ext}}
\newcommand{\arrow}{\mathbin{\rightarrow}}
\newcommand{\fa}[1]{\forall{#1}~.~}
\newcommand{\ex}[1]{\exists{#1}~.~}
\newcommand{\Em}[1]{\textbf{#1}}

\maketitle

\begin{abstract}
We investigate a canonical way of defining bisimilarity of systems when their semantics is given by a coreflection, typically in a category of transition systems. We use the fact, from Joyal et al., that coreflections preserve open morphisms situations in the sense that a coreflection induces a path subcategory in the category of systems in such a way that open bisimilarity with respect to the induced path category coincides with usual bisimilarity of their semantics. We prove that this method is particularly well-suited for systems with quantitative information: we canonically recover the path category of probabilistic systems from Cheng et al., and of timed systems from Nielsen et al., and, finally, we propose a new canonical path category for hybrid systems.
 \end{abstract}

\section{Introduction}

Bisimulations were introduced in \cite{park81}, to express the equivalence of two systems in a way that would reflect not only trace equivalence but also the branching structure of executions. Later, several categorical frameworks were introduced to understand the general theory of bisimilarity. 
In this paper, we adopt the open maps approach of Joyal et al.\ \cite{joyal94, joyal96}, where systems are modeled as objects of a category with a specified subcategory---the \textit{path} category---which models the shapes of executions of the system. Functional bisimulations are modeled as the eponymous \textit{open maps}, morphisms of the category that have a lifting property with respect to the path category.

Many different models of computation can be presented in the open maps framework, leading researchers to compare these categories. Translations from one model to another (for example, Petri nets into event structures, as in \cite{nielsen95}) enable us to compare their expressive power, see for example \cite{vanglabbeek05} for an exhaustive hierarchy for concurrent models. It has been observed that, in many cases, those translations can be made functorial \cite{goubault12, nielsen95}. In particular, Nielsen et al.\ advocated for the use of coreflections, a special kind of adjunction, as a good way to compare the expressive power of two models.

Translations are also used to give semantics for complex system types in terms of well-understood simple systems. Indeed, many models of computation, e.g.\ CCS \cite{milner80}, timed systems \cite{bengtsson03}, and hybrid systems \cite{girard07, henzinger00}, are given semantics by translating them into labelled transition systems. Generally speaking, these semantic translations operate by squeezing the complex parts of a system into the state space and/or action space of the labelled transition system (LTS). For example, as we will see in section~\ref{sec:timed}, the actions of LTSs giving the semantics of timed systems consist of both an action of the original timed system and an arbitrary nonnegative real number representing the elapsed time since the last action. An obvious effect of this semantic translation strategy is that even finite systems of sufficiently complex type usually have infinite sized LTSs representing their semantics. To fit these complex system types in the open maps framework, usually a path (sub)category is carefully handcrafted and then a result is presented confirming that two complex systems are bisimilar according to this custom path category if and only if their semantics are bisimilar as LTSs. 

A general strategy for avoiding this involved isolation of a path category is suggested by a general technical observation of Joyal et al.\ in~\cite{joyal96}. They showed that path categories\tikzstyle{mybox} = [draw=black, fill=gray!20, very thick,
    rectangle, rounded corners, inner sep=4pt, inner ysep=6pt]
\begin{wrapfigure}{r}{0.55\textwidth}
\vspace{-4.5mm}
\begin{tikzpicture}[scale=0.6]
    \node[mybox] (qm) at (0,0) {
    \begin{minipage}{0.20\textwidth}
    \centering
    \scriptsize{Quantitative models\\of computation\\(prob., timed, hybrid)}
    \end{minipage}
    };
    
    \node[mybox] (lts) at (7.5,0) {
    \begin{minipage}{0.17\textwidth}
    \centering
    \scriptsize{LTS\\(with observations)}
    \end{minipage}
    };
    
    \node (adj) at (3.9, -0.4) {$\bot$};
    
    \path[->, font=\scriptsize]
        (qm) edge (lts);
    \node at (3.9, 1) {\scriptsize{semantics}};
    \node at (3.9, 0.63) {\scriptsize{+}};
    \node at (3.9, 0.26) {\scriptsize{unfolding}};
        
    \path[left hook->, dashed, bend left = 15, font=\scriptsize]
        (lts) edge node[below]{coreflection} (qm);
        
    \node (bisqm) at (0, -1.8) {\scriptsize{bisimilarity}};
    \node (bislts) at (7.5, -1.8) {\scriptsize{bisimilarity}};
    
    \path[->, font=\scriptsize, dashed]
        (bislts) edge node[below]{gives rise to} (bisqm);
\end{tikzpicture}
\vspace{-6mm}
\end{wrapfigure} can automatically be translated along coreflections in such a way that the coreflection preserves and reflects bisimilarity. Unfortunately, this general strategy is often not applicable to the semantic case---the common semantic translations are usually not coreflections because they create too many unreachable states.


In this paper, we remedy this defects by composing usual translations with unfolding. This operation allows us to remove the unreachable states present in naive semantics translations, but also to use the unique path property of trees to obtain coreflections. This also suggests that the crucial property of LTSs is \emph{accessibility}, as investigated by the first author in \cite{dubut16}. When a category is accessible with respect to a path category, the open morphism situation is well-behaved, in that open bisimilarity is equivalent to the existence of a bisimulation-like relation on runs, and admits an unfolding operation. Furthermore, accessibility is a property preserved by coreflection, which means that the open morphism situation we synthesise using corefection is automatically well-behaved. This allows us to obtain path categories uniformly and automatically for several quantitative system types using new, but bisimilar, semantics. In particular, we recover the path categories of Cheng et al.~\cite{cheng95} and Nielsen et al.~\cite{nielsen99} for probabilistic systems and timed systems, respectively. Additionally, we synthesise a path category for hybrid systems, which we believe to be the first such category.

\textbf{Organization - }
In Section \ref{sec:TS}, we start by recalling the bisimilarity theory of transition systems using open morphisms, and the classical notion of unfolding. We then show how this theory can be extended to transition systems with observations from \cite{girard11}. In Section \ref{sec:coreflection}, we recall the theory of coreflections and their interactions with open morphisms (Theorem \ref{theo:joyal} and Proposition \ref{theo:dubut}). Finally, in Sections \ref{sec:proba}, \ref{sec:timed} and \ref{sec:hybrid}, we study three examples of quantitative systems by providing a category for each type, describing a translation with values in transition systems (with observations), proving it is a coreflection, and finally synthesising a path category. In the case of probabilistic systems (resp.~timed systems), we canonically obtain the path category of \cite{cheng95} from Theorem \ref{theo:proba} (resp.~\cite{nielsen99} from Theorem \ref{theo:timed}). On the other hand, the path category of hybrid systems obtained from Theorem \ref{theo:hybrid} is completely new.


\section{Bisimilarity in transition systems}
\label{sec:TS}


\subsection{The category of transition systems and open morphisms}

\begin{wrapfigure}{r}{0.25\textwidth}
\vspace{-4mm}
\begin{tikzpicture}[scale=0.6]
	
	\node (eer) at (-1,0) {};	
	\node (1) at (0,0) {\scriptsize{$q_1$}};
	\node (2) at (1.5,-1) {\scriptsize{$q_2$}};
	\node (4) at (1.5,1) {\scriptsize{$q_4$}};
	\node (3) at (3,-1) {\scriptsize{$q_3$}};
	\node (5) at (3, 1) {\scriptsize{$q_5$}};
	
	\node (1p) at (0, -3.5) {\scriptsize{$q'_1$}};
	\node (2p) at (1.5, -3.5) {\scriptsize{$q'_2$}};
	\node (3p) at (3,-4.5) {\scriptsize{$q'_3$}};
	\node (5p) at (3,-2.5) {\scriptsize{$q'_5$}};
	
	\path[->,font=\scriptsize]
		(-0.7,0) edge (1)
		(1) edge node[below]{$a$} (2)
		(1) edge node[above]{$a$} (4)
		(2) edge node[below]{$b$} (3)
		(4) edge node[above]{$c$} (5)
		
		(-0.7,-3.5) edge (1p)
		(1p) edge node[below]{$a$} (2p)
		(2p) edge node[below]{$b$} (3p)
		(2p) edge node[above]{$c$} (5p);
		
	\path[->, dotted, bend left = 20]
		(1) edge (1p)
		(2) edge (2p)
		(4) edge (2p)
		(3) edge (3p)
		(5) edge (5p);
			
\end{tikzpicture}
\vspace{-5mm}
\end{wrapfigure}
From now on fix an alphabet $\Sigma$. A \textbf{(labelled) transition system} is a triple $(S,i,\Delta)$ with $S$ a set of \textbf{states}, $i \in S$ the \textbf{initial state}, and $\Delta \subseteq S\times\Sigma\times S$ the \textbf{transition relation}. A \textbf{morphism of transition systems} $f$ from $T = (S,i,\Delta)$ to $T' = (S',i',\Delta')$ is a function $\map{f}{S}{S'}$ such that for every $(s,a,s') \in \Delta$, $(f(s),a,f(s')) \in \Delta'$. Transition systems and morphisms of transition systems form a category, which we denote by $\ts$. For instance, the diagram on the right presents two transition systems $T_u$ (upper) and $T_d$ (lower) and a morphism $f$ of transition systems from $T_u$ to $T_d$ (dotted arrows). The small arrow into $q_1$ (resp. $q_1'$) represents the initial state of $T_u$ (resp. $T_d$).

There is a special class of transition systems called \Em{finite linear
systems}. A finite linear system of a word $a_1\cdots a_n\in\Sigma^*$ is a transition system specified by the following diagram:
$$
  L(a_1\cdots a_n)=\arrow 0\xrightarrow{~a_1~}1\xrightarrow{~a_2~}\cdots\xrightarrow{~a_{n}~} n,\quad L(\epsilon)=\arrow 0
$$
We write $\lin$ for the full subcategory of $\ts$ consisting of all finite linear systems $\{L(w)~|~w\in\Sigma^*\}$. 
Finite linear systems are used to represent \Em{transition paths} inside a transition system $T$---these paths are simply morphisms $p:P\arrow T$ in $\ts$ from some $P\in\lin$. An important prerequisite to the notion of open maps is the ordering on paths called \textbf{path extension}: given a morphism between finite linear systems $\map eP{P'}$ and a path $\map pPT$, an \textbf{extension of $p$ along $e$} is a morphism $\map{p'}{P'}T$ such that $p'\circ e=p$. We write $\Ext(p,e)$ for the set of extensions of $\map pPT$ along $\map eP{P'}$.

We now consider the interaction between morphisms of transition systems and path extensions. Fix a morphism $\map fT{T'}$ in $\ts$. This maps a transition path $p$ inside $T$ to the path $f\circ p$ inside $T'$. We define a family of binary relations $R_f^P=\{(p,f\circ p)~|~\map pPT\}$ indexed by objects $P$ in $\lin$, which exhibits that $T'$ can simulate path extensions in $T$ in the following sense (to show it, let $q'=f\circ q$):
\begin{equation}
  \label{eq:simu}
  \fa{e\in\lin(P,P')}
  (p,p')\in R_f^P\implies
  \fa{q\in \Ext(p,e)}\ex{q'\in \Ext(p',e)}(q,q')\in R_f^{P'}.
\end{equation}

More interesting is the situation when $R_f^P$ is \Em{bisimulation-like}, that is, it also enjoys the symmetric version of \eqref{eq:simu} (this is what is called (strong) path bisimulation in \cite{joyal96}):
\begin{equation}
  \label{eq:open1}
  \fa{e\in\lin(P,P')}
  (p,p')\in R_f^P\implies
  \fa{q'\in \Ext(p',e)}\ex{q\in \Ext(p,e)}(q,q')\in R_f^{P'}.
\end{equation}
\begin{wrapfigure}{r}{0.3\textwidth}
\vspace{-6mm}
\centering
  \begin{tikzpicture}[scale=0.75]	

    \node (B') at (0,0) {$P'$};
    \node (B) at (0,1.5) {$P$};
    \node (S) at (3,0) {$T'$};
    \node (T) at (3,1.5) {$T$};
    
    \path[->,font=\scriptsize]
    (B) edge node[left]{$e$} (B')
    (B') edge node[below]{$q'$} (S)
    (T) edge node[right]{$f$} (S)
    (B) edge node[above]{$p$} (T);
    
    \path[->, font = \scriptsize, dotted]
    (B') edge node[above]{$q$} (T);	
    
  \end{tikzpicture}
    \vspace{-6mm}
\end{wrapfigure}
This says that $R_f^P$ also witnesses that $T$ can simulate path extensions in $T'$. By unfolding the definition of $R_f^P$ and $\Ext$, \eqref{eq:open1} is actually equivalent to the following \Em{lifting property}: in the figure on the right, for any morphism $e\in\lin$ and $p,q'\in\ts$ making the square commute, there exists a morphism $q$ in $\ts$ making two triangles commute. When $f:T\arrow T'$ satisfies this lifting property, the relation $R_f^P$ witnesses that the transition systems $T$ and $T'$ can mutually simulate path extensions. We call such morphisms $\lin$-\Em{open}, and adopt them as a primitive form of bismulation between two transition systems. In the category $\ts$, a morphism $f:T\arrow T'$ is $\lin$-open if and only if the graph of $f$ is a strong bisimulation relation between $T$ and $T'$. 

To represent general relational bisimulations, it suffices to combine two $\lin$-open morphisms as a \Em{span} $T\leftarrow R\rightarrow T'$. We call such spans \Em{$\lin$-open bisimulation} from $T$ to $T'$. This subsumes usual bisimulation relations between transition systems: for any strong bisimulation $R$ between $T$ and $T'$, its projection legs $T\leftarrow R\rightarrow T'$ are a span of $\lin$-open morphisms.

 
\subsection{Open morphisms and open bisimulation}

The definition of lifting property, open morphisms, and open bisimulations introduced in the previous section make sense in a more general situation abstracting the underlying categories $\ts$ and $\lin$.
By a \textbf{categorical model}, we mean a category $\M$ (of \textbf{systems and functional simulations}, much as $\ts$) and its subcategory $\PP$ (of \textbf{execution shapes and shape extensions}, much as $\lin$). In this framework, a morphism $\map{f}{T}{T'}$ of $\M$ is said to be \textbf{$\PP$-open} if in the figure above, $f$ enjoys the same lifting property stated as before: for any morphism $p,q'$ in $\M$ and $e$ in $\PP$ making the square commute, there exists a morphism $q$ in $\M$ making two triangles commute. A span $T\leftarrow R\rightarrow T'$ of $\PP$-open morphisms is called \Em{$\PP$-open bisimulation} from $T$ to $T'$. We then say that $T$ and $T'$ in $\M$ are \Em{$\PP$-open bisimilar} if there is a $\PP$-open bisimulation from $T$ to $T'$.

\subsection{Unfolding}
\label{subsec:unfolding}

Let $\map I\PP \M$ be a categorical model, and assume a certain cocompleteness (called \Em{$\PP$-accessibility}, formulated by the first author in \cite{dubut16} and repeated below) on $\M$. Using the colimits provided by the $\PP$-accessibility, we can \Em{unfold} any system $X$ into a tree-like system that is bisimilar to $X$:
$$U(X)=\mathrm{colim}(\mapa{\pi}{I\downarrow X}{\mapa I\PP\M}),$$
where $\pi$ is the canonical projection from the comma category $I\downarrow X$.
We call any such colimit a \Em{tree-like system}, and write $\trcm$ for the full subcategory of $\M$ consisting of tree-like systems. The following proposition shows that it is harmless to consider the unfolding instead of the system itself, as long as we are interested in bisimilarity.


\begin{proposition}[\cite{dubut16}]
  When $\M$ is $\PP$-accessible (which is explained below), we have:
  \begin{enumerate}
  \item The mapping $X\mapsto U(X)$ extends to a \Em{coreflection} (Section \ref{sec:coreflection}) $\map{U}{\M}{\trcm}$.
  \item The unit $\map{\unf_X}{U(X)}{X}$ of the coreflection is open.
  \item $X$ and $Y$ are $\PP$-open bisimilar iff $U(X)$ and $U(Y)$ are $\PP$-open bisimilar.
  \end{enumerate}
\end{proposition}

We say that $\M$ is \Em{$\PP$-accessible} if 1) for any nonempty diagram in $\PP$, its $\M$-colimit exists, and 2) for any non-empty diagram $D:\D\arrow\PP$, its $\M$-colimit $(Z,\eta)$ satisfies the following factorization property: for any $a:P\arrow Z$ in $\M$ with $P\in\PP$, there exists $d\in\D$ and $e:P\arrow Dd$ such that $\eta_d\circ e=a$.

\begin{example}
In the categorical model $\maps{\lin}{\ts}$ of transition systems, tree-like systems are a well-known concept. The category $\textbf{Tr(}\ts,\lin\textbf{)}$ of tree-like systems is isomorphic to the category $\tr$ of synchronization trees \cite{winskel84}, and the unfolding of $T = (S, i , \Delta)$ is given by the following tree: its states are runs of $T$, that is, sequences $i = q_0 \xrightarrow{a_1} \ldots \xrightarrow{a_n} q_n$, with $(q_{i-1}, a_i, q_i) \in \Delta$, its initial state is the singleton sequence $i$, and its transitions are $$(q_0 \xrightarrow{a_1} \ldots \xrightarrow{a_n} q_n, a_{n+1}, q_0 \xrightarrow{a_1} \ldots \xrightarrow{a_n} q_n \xrightarrow{a_{n+1}} q_{n+1}).$$
The morphism $\unf_T: U(T) \to T$ maps a run to its ending state.
\end{example}

\subsection{Transition systems with observations}
\label{subsec:observation}

One usual semantics of hybrid systems uses an extension of transition systems with observations \cite{girard11}. Fix a pseudometric space $(\OO, d)$, called the \textbf{observation space}. A \textbf{transition system with observations} is a tuple $(S,i,\Delta,\omega)$ where, $(S,i,\Delta)$ is a transition system and $\map{\omega}{S}{\OO}$ is an \textbf{observation} function. An \textbf{$\epsilon$-bounded morphism} $\map{f}{(S,i,\Delta,\omega)}{(S',i',\Delta',\omega')}$ is a morphism between the underlying transition systems such that $d(\omega(s),\omega'(f(s))) \leq \epsilon$ for every $s \in S$. A \textbf{bounded morphism} is an $\epsilon$-bounded morphism for some $\epsilon \geq 0$. Transition systems with observations and bounded morphisms form a category that we denote by $\tso$. The evident forgetful functor $\maps \tso\ts$ is denoted by $W$. We denote by $\lino$ (resp.~$\tro$) the full subcategory of $\tso$ consisting of systems whose underlying transition system is finite linear (resp.~is a synchronization tree).

\begin{proposition}
  A (bounded) morphism $\map f T {T'}$ in $\tso$ is $\lino$-open if and only if $\map{Wf}{WT}{WT'}$ in $\ts$ is $\lin$-open.
\end{proposition}


This defines a notion of open bisimilarity as the existence of a span of bounded open morphisms. We say that two transition systems with observations are \textbf{$\epsilon$-open bisimilar} if there is a span consisting of an $\epsilon_1$-bounded open morphism and an $\epsilon_2$-bounded open morphism, with $\epsilon_1+\epsilon_2 = \epsilon$. This coincides with the usual notion of approximate bisimilarity \cite{girard11}: an \textbf{$\epsilon$-approximate bisimulation} between two transition systems with observations is a strong bisimulation $R$ between the underlying transition systems such that:
$$(s,s') \in R \Rightarrow d(\omega(s),\omega'(s')) \leq \epsilon.$$

\begin{proposition}
\label{prop:approximate}
The following are equivalent for two systems with observations $T$ and $T'$:
\begin{enumerate}
	\item $T$ and $T'$ are $\epsilon$-open bisimilar.
	\item There is an $\epsilon$-approximate bisimulation between $T$ and $T'$.
	\item There is a span consisting of a $0$-bounded open morphism and an $\epsilon$-bounded open morphism between $T$ and $T'$.
\end{enumerate}
\end{proposition}

The category $\tso$ is not $\lino$-accessible because some colimits do not exist, so we do not automatically have a notion of unfolding. But we can do something by using the unfolding of the underlying transition system. Define the \textbf{unfolding} $V(T)$ of a transition system with observations as the object of $\tro$ whose underlying transition system is the unfolding in $\ts$ of the underlying transition system of $T$ and whose observation is defined as
$$\omega(q_0 \xrightarrow{a_1} \ldots \xrightarrow{a_n} q_n) = \omega(q_n).$$

\begin{proposition}
We have the following properties:
\begin{enumerate}
	\item The mapping $T\mapsto V(T)$ extends to a coreflection functor $\map{V}{\tso}{\tro}$.
	\item The unit $\map{\unf_X}{V(X)}{X}$ is $0$-bounded open.
	\item $T$ and $T'$ are $\epsilon$-approximate bisimilar iff $V(T)$ and $V(T')$ are $\epsilon$-approximate bisimilar.
\end{enumerate}
\end{proposition}

\section{Reflecting bisimilarity through coreflections}
\label{sec:coreflection}

A functor $\map{F}{\M'}{\M}$ is a \textbf{coreflection} if it is the right adjoint of a fully-faithful functor, or equivalently, if it is a right adjoint and the unit of this adjunction is an isomorphism. Concretely, $F$ is a coreflection if there are:
(1) a functor $\map{\iota}{\M}{\M'}$, (2) a natural transformation $\natr{\epsilon}{\iota\circ F}{\text{Id}_{\M'}}$ called \textbf{the counit}, (3) a natural isomorphism $\natr{\eta}{\text{Id}_\M}{F\circ\iota}$ called \textbf{the unit}
and these satisfy for every $X \in \M$ and $X' \in \M'$:
$$\epsilon_{\iota(X)}\circ\iota(\eta_X) = \text{id}_{\iota(X)} \text{~~~~and~~~~} F(\epsilon_{X'})\circ\eta_{F(X')} = \text{id}_{F(X')}.$$

In the case where $F$ gives the semantics of systems in $\M'$ in terms of systems in $\M$, the existence of $\iota$ and of the natural isomorphism $\eta$ tells us that $\M$ acts mostly like a subcategory of $\M'$. We will see later that the existence of the natural transformation $\epsilon$ mainly boils down to the fact that $\M$ is expressive enough to allow us to encode all the additional features from $\M'$. In the case of quantitative systems, in particular timed and hybrid systems, this will mean that $\iota$ is clever enough to encode any possible clock or any possible dynamics.

A crucial property of coreflections we use in the following is that they reflect open morphisms situations:

\begin{theorem}[Corollary 7 of \cite{joyal96}]
\label{theo:joyal}
Let $\maps\PP\M$ be a categorical model, and let $\map{F}{\M'}{\M}$ be a coreflection with fully-faithful left adjoint $\map{\iota}{\M}{\M'}$. Then for every $X$ and $Y \in \M'$, $FX$ and $FY$ are $\PP$-bisimilar iff $X$ and $Y$ are $\iota\PP$-bisimilar.
\end{theorem}

Our plan is then the following. We start with a translation functor $\map{F}{\M'}{\M}$. We know $\M$ well, typically it is a category of transition systems (with observations). In particular, we know an open morphism characterization of bisimilarity in $\M$, given by a path category $\PP$. Then, if the bisimilarity in $\M'$ is defined as the bisimilarity in $\M$ through $F$, this result tells us that it is equivalent to define it directly in $\M'$ using $\iota \PP$-open morphisms. Note the construction of the path category $\iota\PP$ is canonical since adjoints are unique up to isomorphism. Moreover, with the following:

\begin{proposition}[\cite{dubut16}]
\label{theo:dubut}
When $\map{F}{\M'}{\M}$ is a coreflection and $\M$ is $\PP$-accessible, then $\M'$ is $\iota\PP$-accessible.
\end{proposition}
we have that, in those cases, the categorical model $\maps{\iota\PP}{\M'}$ admits a notion of unfolding and an equivalent characterization of $\iota\PP$-bisimilarity using bisimulations \cite{dubut16}.

\section{Probabilistic systems}
\label{sec:proba}

We start with probabilistic systems as studied by Cheng et al.~in \cite{cheng95}, where they model the probabilistic bisimulations of \cite{larsen91} using open morphisms. We first give a translation from probabilistic systems to transition systems consisting of ``forgetting'' probabilities. We prove that this forms a coreflection and that we canonically recover the path category from \cite{cheng95}.

\subsection{The model}

\begin{wrapfigure}{r}{0.33\textwidth}
\vspace{-4mm}
\begin{tikzpicture}[scale=0.8]
	
	\node (eer) at (-1,0) {};	
	\node (1) at (0,0) {\scriptsize{$q_1$}};
	\node (3) at (2,-1) {\scriptsize{$q_3$}};
	\node (2) at (2,1) {\scriptsize{$q_2$}};
	\node (4) at (4,0) {\scriptsize{$q_4$}};
	
	\draw[->] (-0.7, 0) -- (1);
	\path[->, bend left = 30] (1) edge (2);
	\path[->, loop below] (1) edge (1);
	\path[->, bend right = 10] (1) edge (4);
	\path[->, bend right = 20] (2) edge (3);
	\path[->, bend left = 20] (2) edge (4);
	\path[->, bend left = 20] (4) edge (2);
	\path[->, bend left = 10] (4) edge (3);
	
	\node (t1) at (0.2, 0.9) {\scriptsize{$a, \frac{1}{2}$}};
	\node (t2) at (3.5, 0.9) {\scriptsize{$b, 0$}};
	\node (t3) at (0, -0.9) {\scriptsize{$a, \frac{1}{3}$}};
	\node (t4) at (3.5, -0.8) {\scriptsize{$a, \frac{1}{4}$}};
	\node (t5) at (2.4, 0.25) {\scriptsize{$a, \frac{1}{3}$}};
	\node (t6) at (1.5, 0.4) {\scriptsize{$b, \frac{1}{2}$}};
	\node (t7) at (1, -0.4) {\scriptsize{$b, 1$}};

\end{tikzpicture}
\vspace{-5mm}
\end{wrapfigure}
We start by describing a slight modification of partial probabilistic systems from \cite{cheng95}, to avoid the unnecessary use of hyperreals. A \textbf{partial probabilistic system} is a quadruple $(S, i, Supp, \mu)$ where $S$ is the set of \textbf{states}, $i$ is the \textbf{initial state}, $Supp \subseteq S\times\Sigma\times S$ is the \textbf{support relation}, and $\map{\mu}{Supp}{[0,1]}$ is the \textbf{transition distribution}. These data are required to satisfy that for every $s \in S$, and $a \in \Sigma$, $\{t \mid \mu(s,a,t) > 0\}$ is finite and $\sum\limits_{(s,a,t) \in Supp} \mu(s,a,t) \leq 1$. The adjective ``partial'' reflects the fact that these transition systems actually have transition \textit{sub}distributions. Note that $(S, i, Supp)$ is a transition system, and we will call it the underlying transition system.

A key point to observe is that we are distinguishing between having a transition $(s,a,t)$ with probability $0$ and not having a transition at all by our use of the support relation. This distinction prompted the use of the hyperreals in \cite{cheng95}, where the difference between probability 0 and infinitesimal probability provided the same effect. 

A \textbf{morphism of partial probabilistic systems} from $(S, i, Supp, \mu)$ to $(S', i', Supp', \mu')$ is a morphism $f$ between the underlying transition systems such that for every $(s,a,t) \in Supp$:
$$\sum\limits_{(s,a,t') \in Supp \mid f(t') = f(t)} \mu(s,a,t') \leq \mu'(f(s),a,f(t)).$$
We denote the category of partial probabilistic systems and morphisms by $\prob$.

\subsection{The coreflection}

There is an obvious translation from probabilistic systems to transition systems obtained by forgetting the transition distribution, i.e., mapping a partial probabilistic system $(S,i,Supp,\mu)$ to its underlying transition system $(S,i,Supp)$. We name this mapping $F$.

\begin{theorem}
\label{theo:proba}
  The mapping $F$ extends to a coreflection functor $\map{F}{\prob}{\ts}$.
\end{theorem}


The fully faithful functor $\iota$ from $\ts$ to $\prob$ assigns the everywhere $0$ transition distribution to the support relation. To prove that this is a coreflection, we need to construct the counit morphism $\map{\epsilon_T}{(\iota\circ F)T}{T}$ for every partial probabilistic system $T$. Since $F$ forgets the transition distribution from $T$ and the morphism $\epsilon_T$ cannot increase the probabilities from the domain to the codomain, the transition distribution on $(\iota\circ F)T$ must be below any possible distribution on $T$. Therefore, $\iota$ must assign the everywhere $0$ distribution, and note that the uniqueness up to isomorphism of adjoints reflects this single possible choice of distribution. (This is also why Cheng et al.~fixed it to be an infinitesimal in \cite{cheng95}).

\subsection{The path category}

The path category we obtain from this coreflection consists of those partial probabilistic systems whose underlying transition system is finite linear and whose transition distribution is $0$ everywhere. This path category is equivalent to that of \cite{cheng95}, modulo the fact that infinitesimal probability in their work corresponds to probability $0$ in ours. 

\section{Timed systems}
\label{sec:timed}

In this section, we start by recalling a common translation from timed systems to transition systems as described in \cite{nielsen99}, following the work from \cite{alur90, bengtsson03}. This does not form a coreflection, but we prove that it does after composing it with the unfolding functor. We then prove that we recover the path category from \cite{nielsen99} in a canonical way.

\subsection{The model}

\begin{wrapfigure}{r}{0.4\textwidth}
\vspace{-6mm}
\begin{tikzpicture}[scale=1]
	
	\node (eer) at (-1,0) {};	
	\node (1) at (0,0) {\scriptsize{$q_1$}};
	\node (3) at (2,-1) {\scriptsize{$q_3$}};
	\node (2) at (2,1) {\scriptsize{$q_2$}};
	\node (4) at (4,0) {\scriptsize{$q_4$}};
	
	\draw[->] (-0.7, 0) -- (1);
	\path[->, bend left = 30] (1) edge (2);
	\path[->] (2) edge (3); 
	\path[->] (4) edge (2); 
	\path[->, bend left = 20] (3) edge (4);
	
	\node (t1) at (0.1, 1.2) {\scriptsize{$a, y:= 0$}};
	\node (t1p) at (0.1, 0.9) {\scriptsize{$x,y \in \rea$}};

	\node (t2) at (3.6, 1.3) {\scriptsize{$b, y := 0$}};
	\node (t2p) at (3.6, 1.0) {\scriptsize{$y \in ]1,2]$}};
	\node (t2pp) at (3.6, 0.7) {\scriptsize{$x \in \rea$}};

	\node (t4) at (3.3, -0.6) {\scriptsize{$a, \varnothing := 0$}};
	\node (t4p) at (3.3, -.9) {\scriptsize{$x \in [0, 1[$}};
	\node (t4pp) at (3.3, -1.2) {\scriptsize{$y \in \rea$}};

	\node (t6) at (1.3, 0.3) {\scriptsize{$b, \varnothing := 0$}};
	\node (t6p) at (1.3, 0) {\scriptsize{$y \in [1,1]$}};
	\node (t6pp) at (1.3, -0.3) {\scriptsize{$x \in \rea$}};

\end{tikzpicture}
\vspace{-8mm}
\end{wrapfigure}
Given a set $C$, we write $G_C$ for the set of subsets of $\rea^C$ of the form $\prod_{c \in C} I_c$, where $I_c$ is a non-negative non-empty interval. An element of $G_C$ is called a \textbf{guard}. The guards will give conditions on the values of the clocks of the system to allow a transition to be used.  A \textbf{timed transition system} is a quadruple $(S,i,C,\Delta)$ with $S$ the set of \textbf{states}, $i \in S$ the \textbf{initial state}, $C$ the set of \textbf{clocks}, and $\Delta \subseteq S\times\Sigma\times2^C\times G_C\times S$ the \textbf{transition relation}. The $2^C$ component of the transition relation describes which clocks are reset after the transition is completed.  A \textbf{morphism of timed transition systems} $(f,g)$ from $T = (S,i,C,\Delta)$ to $T' = (S',i',C',\Delta')$ is a function $\map{f}{S}{S'}$ and a function $\map{g}{C'}{C}$ such that for every $(s,a,R,\prod_{c \in C} I_c,s') \in \Delta$, there is a transition $(f(s),a,R', \prod_{c' \in C'} I'_{c'}, f(s')) \in \Delta'$ with $R' = g^{-1}(R)$ and $\forall c'\in C',\, I_{g(c')} \subseteq I'_{c'}$.

Remember that a morphism is meant to be a functional simulation, so in particular it should map runs into runs. A morphism of timed systems is then composed of two parts: a state part and a clock part. The state part is as for usual transition systems, while the clock is contravariant. This means that a morphism can essentially duplicate or forget clocks, but cannot create clocks. The reason is that, with guards, clocks gives conditions to satisfy to use a transition. Creating clocks would create new conditions that can fail in the codomain, which may forbid some runs that should not be forbidden. The two conditions then means that $c' \in C'$ acts at least as $g(c') \in C$, that is, when $g(c')$ is reset, $c'$ is reset, and guards on $c'$ are weaker than the guards on $g(c')$. Timed transition systems and morphisms of timed transition systems form a category, which we denote by $\tts$.

\subsection{The coreflection}

One usual translation of timed transition systems into labelled transition systems is given by \textbf{configurations}, where the state of the labelled transition system consists of the current state and the current values of the clocks in the timed transition system, and LTS transitions are given by an action and a time. Concretely, given a timed transition system $T = (S,i,C,\Delta)$, define the labelled transition system on the alphabet $\Sigma\times\mathbb{R}_{>0}$: $\Theta T = (S\times \rea^C, (i,\tilde{0}), \Gamma)$ where $\tilde{0}$ is the valuation that maps every clock to $0$, and $\Gamma$ is defined by $((s,\nu),(a,t),(s',\nu')) \in \Gamma$ if and only if there is $(s,a,R,\prod_{c \in C} I_c,s') \in \Delta$ such that:
$$\forall c \in C,\,\nu(c) + t \in I_c \text{~~~~and~~~~} \nu' = 
\left\{
\begin{array}{l c}
	\nu(c) + t	&	\text{if } c\notin R \\
    	0     		&	\text{if } c\in R
\end{array}
\right.$$

This extends to a functor $\map{\Theta}{\tts}{\tsrea}$, which is not quite a coreflection because the transition system $\Theta T$ has many unreachable states, namely impossible clock configurations. To find a coreflection, we consider $G = U\circ\Theta$ instead, where $U$ is the unfolding functor from Section \ref{subsec:unfolding}, which does not change the semantics modulo bisimilarity.

\begin{theorem}
\label{theo:timed}
The functor $\map{G}{\tts}{\trrea}$ is a coreflection.
\end{theorem}


The fully faithful functor $\iota$ from $\trrea$ to $\tts$ is quite technical: as observed in \cite{nielsen99}, getting a TTS from a TS, even finite linear, boils down to the definition of the clocks. If we come back to the description of coreflections from Section \ref{sec:coreflection}, for every TTS $T$ we need a counit $\epsilon_T$ from $(\iota\circ G)T$ to $T$. Remember that a morphism of TTS is composed of a state part (which is covariant) and of a clock part (which is contravariant). In particular, this means that for every clock of $T$, there must be a clock of $(\iota\circ G)T$ with the same resets. Concretely, given a synchronization tree on the alphabet $\Sigma\times\mathbb{R}_{>0}$ $T = (S,i,\Delta)$, we define the TTS $\iota T = (S, i, C, \Gamma)$ where:
\begin{itemize}
	\item $C = 2^\Delta$,
	\item $\Gamma$ is the set of transitions defined as follows. For every transition $(s, (a,t),s') \in \Delta$, if $i = q_0 \xrightarrow{(a_1,t_1)} \ldots \xrightarrow{(a_n,t_n)} q_n = s$ is the unique run from the initial state to $s$ in $T$, we have a transition $(s, a, R, \prod\limits_{U \subseteq C} \{t_U\}, s') \in \Gamma$ with:
	\begin{itemize}
		\item $R = \{U \subseteq \Delta \mid (s, (a,t), s') \in U\}$,
		\item $t_U = t + \sum\limits_{j = i_U +1}^{n} t_j$ where $i_U = \max\{j \leq n \mid (q_{j-1}, (a_j,t_j), q_j) \in U\}$.
	\end{itemize}
\end{itemize}
The intuition is that $C$ encodes all the possible behaviors of a clock on $T$, that is, all the possible points where a clock can be reset: a clock of $T$ is encoded by the set of transitions on which it is reset. The guard then ensures that the values of the clocks are coherent with the last time it has been reset. The fact that $T$ is a tree is important as we crucially use that there is a unique path from the initial state, and particularly, a unique sequence of reset times. As described previously, the interesting part is the definition of the clock part of $\epsilon_T$. Given a TTS $T = (S, i, C, \Gamma)$, a transition of $GT$ is of the form $$((i,\tilde{0}) = (s_0, \nu_0) \xrightarrow{(a_1,t_1)} \ldots \xrightarrow{(a_n,t_n)} (s_n,\nu_n), (a_{n+1}, t_{n+1}), ~~~~~~~~~~$$
$$~~~~~~~~~~(s_0, \nu_0) \xrightarrow{(a_1,t_1)} \ldots \xrightarrow{(a_n,t_n)} (s_n,\nu_n) \xrightarrow{a_{n+1}, t_{n+1}} (s_{n+1}, \nu_{n+1}))$$
A clock $c \in C$ is then encoded by the set of such transitions such that $\nu_{n+1}(c) = 0$, that is, the set of transitions in $GT$ that corresponds to transitions of $T$ where $c$ has been reset.

\subsection{The path category}

Describing the path category in $\tts$ then boils down to looking at the image of a finite linear transition system $T$ by $\iota$. $\iota T$ has the same states as $T$ and has the same number of transitions between two given states. Consequently, this is a finite linear TTS. The clocks are subsets of transitions of $T$, or equivalently, a subset of $\{1, \ldots, n\}$. $R_i$ is then the set of subsets $U$ of $\{1, \ldots, n\}$ such that $i \in U$. $G_i$ tells us that the value of $U$ must be equal to the time since the last time a transition of $T$ belonged to $U$. In summary, this path category is isomorphic to the one given by \cite{nielsen99}, and we have shown how to canonically recover it.


\section{Hybrid systems}
\label{sec:hybrid}

In this section, we study the case of hybrid systems and their translation to transition systems with observations. We prove that, again, when composing this translation with the unfolding functor from section \ref{subsec:observation}, we get a coreflection. This produces a new canonical path category on hybrid systems.

\subsection{The model}

\newcommand{\modebox}[6]
{
  \tikzmath{
    \a=1.5; 
    \aa=.3; 
    \b=.7; 
    \c=.7; 
    \x=3.5; 
    \w=\x/2; 
    \y=\a+\b+\c; 
  }
  \node (#6) at (0,0) [draw,minimum width=\x cm,minimum height=\y cm,anchor=south west] {};
  \node [above=0cm of #6] {#1};

  \node (obs) at (0,0) [draw,minimum width=\x cm,minimum height=\c cm,anchor=south west] {#5};
  \node (obstip) at (0, \c) [draw,minimum width=\aa cm,minimum height=\aa cm, anchor=north west] {Ob};
  
  \node (inv) at (0, \c) [draw,minimum width=\x cm,minimum height=\b cm,anchor=south west] {#4};
  \node (invtip) at (0, \c+\b) [draw,minimum width=\aa cm,minimum height=\aa cm, anchor=north west] {Iv};

  \node (diff1) at (0, \c + \b) [draw,minimum width=\w cm,minimum height=\a cm,anchor=south west] {#2};

  \node (diff1tip) at (0, \y) [draw,minimum width=\aa cm,minimum height=\aa cm, anchor=north west] {1};

  \node (diff2) at (\x / 2, \c + \b) [draw,minimum width=\w cm,minimum height=\a cm,anchor=south west] {#3};

  \node (diff2tip) at (\x / 2, \y) [draw,minimum width=\aa cm,minimum height=\aa cm, anchor=north west] {2};
}

In this section, we present hybrid systems similarly to definitions in the literature, e.g.~\cite{henzinger00}. A \textbf{hybrid system} is a undecuple $$(M, I, (n_i)_{i\in I}, E, (G_e)_{e\in E}, (R_{e,i})_{e\in E, i\in I}, (F_{m,i})_{m\in M, i\in I}, (I_m)_{m\in M}, m_0, \sigma_0, o)$$
\begin{wrapfigure}{r}{.6\textwidth}
\vspace{-4mm}
\begin{tikzpicture}[scale=0.6,transform shape]
  \modebox{$M_1$}{$\dot x=K_1$}{$\dot y=-K_2$}{$y\ge x-\alpha$}{$y-x$}{mone}
  \begin{scope}[shift={(7,0)}]
    \modebox{$M_2$}{$\dot x=-L_1$}{$\dot y=L_2$}{$x\ge y-\beta$}{$x-y$}{mtwo}
  \end{scope}

  \draw [-{>[scale=2.0]}]
  (mone.west) + (-2.0,0)
  to
  node [above, text width=3cm, align=center] {$x:=5$\\$y:=10$}
  (mone);
  
  \draw [bend left=15,-{>[scale=2.0]}]
  (mone) to
  node [above,text width=3cm,align=center] {$x:=x,~y:=y$\\$y\le x$}
  node [below] {switch}
  (mtwo);
  
  \draw [bend left=15,-{>[scale=2.0]}]
  (mtwo) to
  node [below,text width=3cm,align=center] {$x:=x,~y:=y$\\$x\le y$}
  node [above] {switch}
  (mone);
\end{tikzpicture}
\vspace{-6mm}
\end{wrapfigure}
where, $M$ is a set of \textbf{modes}, $I$ is a set of \textbf{subsystems}, $n_i \in \mathbb{N}$ is the \textbf{dimension} of the subsystem $i$, $E \subseteq M\times\Sigma\times M$ is a set of \textbf{events}, $G_e \subseteq \prod_{i\in I} \mathbb{R}^{n_i}$ is a \textbf{guard predicate}, $\map{R_{e,i}}{\mathbb{R}^{n_i}}{\mathbb{R}^{n_i}}$ is a \textbf{reset function}, $\map{F_{m,i}}{\mathbb{R}\times\mathbb{R}^{n_i}}{\mathbb{R}^{n_i}}$ is a continuous and locally Lipschitz in the second argument \textbf{flow function}, $I_m \subseteq \prod_{i\in I} \mathbb{R}^{n_i}$ is an \textbf{invariant predicate}, $m_0 \in M$ is the \textbf{initial mode}, $\sigma_0 = (\sigma_{i,0})_{i \in I} \in \prod_{i \in I} \mathbb{R}^{n_i}$ is the \textbf{initial valuation}, and $\map{o}{M\times\prod_{i\in I} \mathbb{R}^{n_i}}{\OO}$ is the \textbf{observation function}.
The above diagram is a graphical representation of a hybrid system. Each box corresponds to a mode of the system, and contains three data on the mode: differential equations for subsystems (tagged 1,2), the invariant (tagged Iv) and the observation function (tagged Ob). The right hand side of the differential equation at each subsystem is given by the flow function of the mode. Each arrow between mode boxes represents an event, and is annotated with a label (e.g. switch), a reset function (e.g. $x:=x,y:=y$) and a guard predicate (e.g. $x\le y$). Finally, the sourceless arrow going into $M_1$ designates the initial mode. This initial mode arrow is also annotated by the initial valuation ($x:=5,y:=10$).

\subsection{The coreflection}

To define the translation from hybrid systems to transition systems, similarly to \cite{girard07,henzinger00}, we need the notions of configurations and runs. A \textbf{configuration} is a pair consisting of a mode together with compatible values of the dynamics variables. Concretely, a configuration of a hybrid system $T$ is a pair $(m,\sigma) \in M\times \prod_{i\in I} \mathbb{R}^{n_i}$. We say that \textbf{$T$ moves from the configuration $(m,\sigma)$ to the configuration $(m', \sigma')$ by doing the action $a$ with time $t$} if the triple $e=(m, a, m')$ is in $E$, and for every $i \in I$, there is a differentiable function $\map{x_i}{[0,t]}{\mathbb{R}^{n_i}}$ such that:
\begin{itemize}
\item for every $s \in [0,t]$, $\dot{x_i}(s) = F_{m,i}(s,x_i(s))$ and $(x_i(s))_{i \in I} \in I_{m}$,
\item $(x_i(0))_{i\in I} = \sigma$ and $(x_i(t))_{i\in I} \in G_e$,
\item $\sigma' = (R_{e,i}(x_i(t)))_{i\in I}$.
\end{itemize}
The intuition is that $T$ has a discrete transition from $m$ to $m'$ with action $a$, and $\sigma'$ is obtained from $\sigma$ by applying the dynamics from mode $m$ for $t$ time units, and by resetting under the conditions given by the invariant and the guard. A \textbf{run} of $T$ is then a sequence denoted by:
$$(m_0, \sigma_0) \xrightarrow{a_1, t_1} \ldots \xrightarrow{a_k,t_k} (m_k, \sigma_k)$$
such that for every $j \in \{1, \ldots, k\}$, $T$ moves from $(m_{j-1}, \sigma_{j-1})$ to $(m_j, \sigma_j)$ by doing the action $a_j$ with time $t_j$.

 Given a hybrid system $T$, we define a transition system with observations on the alphabet $\Sigma\times\rea$ by $K T = (S,i,\Delta, \omega)$ where the states $S$ are configurations, the initial state is $i$ is $(m_0,\sigma_0)$, $((m,\sigma), (a,t), (m',\sigma')) \in \Delta$ if and only if $T$ moves from $(m,\sigma)$ to $(m', \sigma')$ by doing the action $a$ with time $t$, and the observation function is $\omega(m,\sigma) = o(m,\sigma)$.

This translation encodes continuous and discrete transitions all at once. Purely continuous transitions can be represented by an event $\text{id}_m = (m,\tau,m)$ where $\tau$ is a fresh letter to your hybrid system. Guards are then given by $G_{\text{id}_m} = \mathbb{R}^n$ and the reset functions are given by identities. Purely discrete transitions are those labelled by $(a,0)$, where the second component indicates no time has elapsed.

$K$ will not give a coreflection for the same reasons as for timed systems. So, the main idea is, again, to consider $H \triangleq V\circ K$, where $V$ is the unfolding on transition systems with observations from Section \ref{subsec:observation}, which does not affect the semantics modulo approximate bisimulations. $HT$ is then a synchronization tree with observations, whose states are runs of $T$ and whose transitions are extensions of runs with an additional move. To make $H$ into a coreflection, we need to specify a notion of morphism of hybrid systems, which has, as far as we know, never been done before. A \textbf{$\epsilon$-bounded morphism of hybrid systems} is a pair $(f_M, f_I)$ where $\map{f_M}{M}{M'}$ and $\map{f_I}{I'}{I}$ are functions. 
These data must satisfy some conditions:
\begin{itemize}
	\item $f_M$ is a morphism between the underlying transition systems given by modes and events, that is:
	\begin{itemize}
		\item $f_M(m_0) = m_0'$.
		\item if $(m,a,m') \in E$, then $(f_M(m), a, f_M(m')) \in E'$. If $e = (m,a,m')$, we denote by $f_E(e)$ the event $(f_M(m), a, f_M(m'))$.
	\end{itemize}
	\item the subsystems $i' \in I'$ and $f_I(i') \in I$ are the same, that is:
	\begin{itemize}
		\item they have the same dimension: $n_{i'} = n_{f_I(i')}$. 
		\item they have the same dynamics: $F_{m,f_I(i')}(t, x_{i'}) = F'_{f_M(m), i'}(t, x_{i'})$.
		\item they have the same reset functions: $R_{e, f_I(i')} = R'_{f_E(e), i'}$.
	\end{itemize}
	\item We can define the function $\map{f_X}{\prod_{i\in I} \mathbb{R}^{n_i}}{\prod_{i'\in I'} \mathbb{R}^{n_{i'}}}$ by $f_X((x_i)_{i\in I}) = (x_{f_I(i')})_{i' \in I'}$. This function must preserve guards, invariants and the initial valuation, that is:
	\begin{itemize}
		\item $f_X(\sigma_0) = \sigma_0'$, that is, $\sigma'_{i',0} = \sigma_{f_I(i'),0}$.
		\item if $\sigma \in I_m$, then $f_X(\sigma) \in I'_{f_M(m)}$.
		\item if $\sigma \in G_e$, then $f_X(\sigma) \in G'_{f_E(e)}$.
	\end{itemize}
	\item observations are close to each other, that is, for every run $(m_0, \sigma_0) \xrightarrow{a_1, t_1} \ldots \xrightarrow{a_k,t_k} (m_k, \sigma_k)$, $d(o(m_k,\sigma_k),o'(f_M(m_k),f_X(\sigma_k))) \leq \epsilon$.
\end{itemize}
By \textbf{bounded morphism} we mean $\epsilon$-bounded for some $\epsilon \geq 0$. Hybrid systems and bounded functions form a category that we denote by $\hysy$. The intuition underlying the notion of morphisms of hybrid systems is very similar to morphisms of timed systems: it is essentially a morphism of usual transition systems, plus a contravariant part that explicitly marks which subsystem of the domain models which subsystem of the codomain. Again, the contravariance means that we can duplicate or forget dynamics, but not create new ones, for the same reasons as clocks in timed systems.

\begin{theorem}
\label{theo:hybrid}
The functor $H:\hysy\arrow \trorea$ is a coreflection, whose unit and counit are $0$-bounded.
\end{theorem}

This time the tricky part of the adjunction comes from the dynamics. Similarly to the contravariant clock part in timed systems, we have a contravariant subsystem part given by $f_I$. So the subsystem part of the counit $\epsilon_T$ from $\iota\circ HT$ to $T$ tells us that the dynamics of $\iota\circ HT$ must simulate the dynamics of $T$. Let us describe more concretely the functor $\iota$ and the subsystem part of the counit. Given a synchronization tree with observations $T = (S, i, \Delta, \omega)$ in $\trorea$, we define the following hybrid system $\iota T$:
\begin{itemize}
	\item Its modes are states $S$.
	\item Its events are $E = \{(s,a,s') \mid \exists t \in \rea, \, (s, (a,t), s') \in \Delta\}$. Observe that the existing $t$ is necessarily unique, since $T$ is a tree.
	\item Its subsystems $I$ are given by tuples $(n,\sigma, (F_s)_{s\in S}, (R_e)_{e\in E})$ with:
		\begin{itemize}
			\item $\sigma \in \mathbb{R}^n$.
			\item $\map{F_s}{\mathbb{R}\times\mathbb{R}^n}{\mathbb{R}^n}$ continuous and locally Lipschitz in the second variable.
			\item $\map{R_e}{\mathbb{R}^n}{\mathbb{R}^n}$.
		\end{itemize}
	The idea is that they contain all the possible kinds of subsystems, and that they are indexed by their own content.
	\item Before describing guards and invariants, let us describe the dynamics more carefully as follows. Given a state $s$ of $T$, let $\pi$ be the unique run from $i$ to $s$ in $T$ (which is a tree):
	$$i = s_0 \xrightarrow{a_1, t_1} \ldots \xrightarrow{a_k,t_k} s_k = s.$$
	Given additionally $\alpha = (n, \sigma, (F_{s})_{s \in S}, (R_{e})_{e \in E}) \in I$ and $t \in \rea$, define $\sigma_{s, \alpha}(t) \in \mathbb{R}^n$ as $x_{\alpha}^k(t)$, where $x_{\alpha}^j$ is the solution on $[0, t_{j+1}]$ (with the convention that $t_{k+1} = t$) of the equation $\dot{x}(u) = F_{s_j}(u,x(u))$ with the following initial condition:
		\begin{itemize}
			\item $x(0) = \sigma$ if $j = 0$,
			\item $x(0) = R_{(s_{j-1}, a_{j}, s_{j})}(x_{\alpha}^{j-1}(t_{j}))$ if $j > 0$.
		\end{itemize}
	Those $\sigma_{s,\alpha}$ give the values of all possible dynamics by following the times given by the unique run $\pi$.
	\item $G_{(s,a,s')} = \{(\sigma_{s,\alpha}(t))_{\alpha \in I}\}$, where $t$ is the unique time such that $(s, (a,t), s') \in \Delta$.
	\item The invariant $I_s$ is given similarly as the set of all $(\sigma_{s,\alpha}(t))_{\alpha\in I}$, with $t$ such that there is a transition $(s, (a, t'), s') \in \Delta$, with $t \leq t'$. This invariant lets us use the dynamics as long as we need, that is, as long as we still can pass a guard.
	\item $R_{(s,a,s'), (n, \sigma, (F_s)_{s\in S}, (R_e)_{e\in E})} = R_{(s,a,s')}$,
	\item $F_{s, (n,\sigma, (F_s)_{s \in s}, (R_e)_{e \in E})} = F_s$,
	\item $m_0 = i$,
	\item $\sigma_0 = (\sigma)_{(n,\sigma, (F_s)_{s\in S}, (R_e)_{e \in E}) \in I}$,
	\item $o(s, \sigma) = \omega(s)$.
\end{itemize}
Now, to describe the subsystem part of the counit from $\iota \circ H T$ to $T$, for a hybrid system $T$, observe first that the subsystems of $(\iota\circ H) T$ are tuples of the form $$(n, \sigma, (F_s)_{s \text{ state of } HT}, (R_e)_{e \text{ transition of } HT})$$ where $n$ is a natural number, $\sigma \in \mathbb{R}^n$, $(F_s)_{s \text{ state of } HT}$ is a collection of functions indexed by states of $HT$, that is, runs of $T$, and $(R_e)_{e \text{ transition of } HT}$ is a collection of functions indexed by transitions of $HT$, that is, extension of runs. Given a subsystem $i$ of $T$, this subsystem will be encoded by the following tuple:
\begin{itemize}
	\item the natural number is $n_i$,
	\item the vector is $\sigma_i$ where $\sigma_0 = (\sigma_i)_{i\in I}$,
	\item $F_s$, where $s$ is a run of $T$ whose last state is of the form $(m, \sigma)$, is given by $F_{m,i}$,
	\item $R_e$, where $e$ is an extension of runs whose added transition is of the form $(m_1, \sigma_1) \xrightarrow{a,t} (m_2, \sigma_2)$, is given by $R_{(m_1, a, m_2), i}$.
\end{itemize}
In other words, $i$ is modeled essentially by a copy of itself that we can find in $(\iota\circ H) T$.

\subsection{The path category}

The path category is given by those hybrid systems, whose underlying transition system, that is, whose states are modes, and transitions are events, is finite linear, subsystems encodes every possible dynamics and resets, guards allow us to leave a mode precisely at the time given by the transition and the invariants allow us to stay in this mode as long as we need before leaving it.

\section{Conclusion}

In this paper, we have shown how coreflections can be used to canonically define a path category for systems with quantitative information (probabilistic, timed, and hybrid). For the cases of probabilistic and timed systems, we recover the path categories from the literature, except that, in the present paper, they come from a canonical categorical constructions, which makes some of the main theorems from \cite{cheng95, nielsen99} automatic. We also use the same ideas to construct a path category for hybrid systems, which is a novelty. 

As a future work, we would like to understand more clearly the crucial role of the unfolding in those coreflections, and see if this could be a general pattern, for example, for $\PP$-accessible categories. We would also like to use the path category on hybrid systems in practice to help proving that such systems are approximate bisimilar, similarly to what have been done in \cite{nielsen99} for timed systems.

%

\bibliography{coreflections}

\pagebreak
\appendix
\section{Omitted proofs from Section \ref{subsec:observation}}

\begin{proof}[Proof (of Proposition \ref{prop:approximate})]~
\begin{itemize}
	\item[iii) $\Rightarrow$ i)] Obvious.
	\item[i) $\Rightarrow$ ii)] Given such a span $\map{f}{T''}{T}$ and $\map{g}{T''}{T'}$, define the relation $R = \{(f(s''), g(s'')) \mid s'' \in S''\}$. Since $\lino$-open are $\lin$-open, then we know from the case of usual transition system that $R$ is a strong bisimulation. Finally, given $s'' \in S''$, we need to prove that $d(\omega(f(s'')),\omega'(g(s''))) \leq \epsilon$. This is proved as follows:
\begin{center}
\begin{tabular}{rclcl}
    $d(\omega(f(s'')),\omega'(g(s'')))$ & $\leq$ & $d(\omega(f(s'')),\omega''(s'')) + d(\omega''(s''),\omega'(g(s'')))$ & ~~~~~ & (triangular inequality)\\
    & $=$ & $d(\omega(f(s'')),\omega''(s'')) + d(\omega'(g(s'')),\omega''(s''))$ & ~~~~~ & (symmetry)\\
    & $\leq$ & $\epsilon_1 + \epsilon_2$ & ~~~~~ & (boundedness of $f$ and $g$) \\
    & $=$ & $\epsilon$ & ~~~~~ & (assumption)
\end{tabular}
\end{center}
	\item[ii) $\Rightarrow$ iii)] Given an $\epsilon$-approximate bisimulation $R$ between $T$ and $T'$, define $T_R = (R, (i,i'), \Delta_R, \omega_R)$ the following transition system with observations:
	\begin{itemize}
		\item $\Delta_R = \{((s,s'), a, (t,t')) \mid (s,a,t) \in \Delta \wedge (s',a,t') \in \Delta'\}$,
		\item $\omega_R(s,s') = \omega(s)$.
	\end{itemize}
	Define $f$ (resp. $g$) be the first (resp. second) projection. Since $R$ is an approximate bisimulation, then it is a strong bisimulation between the underlying transition systems, so we know, by the case of usual transition systems, that $f$ and $g$ are $\lin$-open and so $\lino$-open. $f$ is $0$-bounded by definition, and $g$ is $\epsilon$-bounded because $R$ is $\epsilon$-approximate.
\end{itemize}
\end{proof}

\section{Omitted proofs on probabilistic systems}

The translation functor is defined as $F(S,i,Supp, \mu) = (S,i,Supp)$. Given a morphism of partial probabilistic systems $f$ from $(S, i, Supp, \mu)$ to $(S', i', Supp', \mu')$, we define $F(f) = f$, which is a morphism of transition systems by definition. This is a functor.

$\map{\iota}{\ts}{\prob}$ is defined on the objects as $\iota(S,i,\Delta) = (S, i, \Delta, 0_\Delta)$, where $0_\Delta$, the function from $\Delta$ to $[0,1]$, equal to $0$ everywhere. Given a morphism of transition systems $f$ from $(S,i,\Delta)$ to $(S', i', \Delta')$, $\iota(f) = f$ is a morphism of partial probabilistic systems since for every $(s,a,t) \in \Delta$:
$$\sum\limits_{(s,a,t') \in \Delta \mid f(t') = f(t)} 0_\Delta(s,a,t') = 0 = 0_{\Delta'}(f(s),a,f(t)).$$

For every transition system $T$, $F\circ\iota(T) = T$, so we can define the unit $\eta_T$ as $\text{id}_T$. This is a natural isomorphism.

For every partial probabilistic system $T = (S,i,Supp,\mu)$, $\iota\circ F(T) = (S,i,Supp,0_{Supp})$. So if we define $\eta_T$ as $\text{id}_S$, this is a morphism of partial probabilistic systems from $\iota\circ F(T)$ to $T$, since for every $(s,a,t) \in Supp$:
\begin{center}
\begin{tabular}{rclcl}
    $\sum\limits_{(s,a,t') \in Supp \mid \text{id}_S(t') = \text{id}_S(t)} 0_{Supp}(s,a,t')$ & $=$ & $0_{Supp}(s,a,t)$ & ~~~~~ & (identity)\\
    & $=$ & $0$ & ~~~~~ & (definition of $0_{Supp}$)\\
    & $\leq$ & $\mu(s,a,t)$ & ~~~~~ & (positivity of $\mu$)
\end{tabular}
\end{center}

The naturality of $\epsilon$ and the two conditions are obvious, since $\eta$ and $\epsilon$ are identity functions.

\section{Omitted proofs on timed systems}

\begin{lemma}
$G$ extends to a functor.
\end{lemma}

\begin{proof}
We have defined $G$ as $U\circ\Theta$, where $U$ is the unfolding functor. To prove it is a functor, we have to prove that $\Theta$ is a functor. Let $(f,g)$ be a morphism of TTS from $T = (S,i,C,\Delta)$ to $T' = (S',i',C',\Delta')$. We want a morphism $\Theta (f,g)$ of TS from $\Theta T = (S\times \rea^C, (i,\tilde{0}), \Gamma)$ to $\Theta T' = (S'\times \rea^{C'}, ('i,\tilde{0}), \Gamma')$. We define $\Theta(f,g)(s,\nu) = (f(s),\nu\circ g)$ ($\nu \in \rea^C$ is identified with a function from $C$ to $\rea$). Let us check that this is a morphism of TS. Let $((s,\nu),(a,t),(s',\nu')) \in \Gamma$ and let $(s,a,R,\prod\limits_{c \in C} I_c,s') \in \Delta$ the associated transition in $T$ in the definition of $\Theta T$. Since $(f,g)$ is a morphism of TTS, there is an associated transition $(f(s),a,R',\prod\limits_{c' \in C'} I'_{c'},f(s'))$ in $T'$. Let us prove that $((f(s),\nu\circ g), (a,t), (f(s'), \nu'\circ g))$ is a transition of $\Theta T'$ using this transition of $T'$. That is, we have to prove that:
\begin{itemize}
	\item for every $c' \in C$, $\nu\circ g(c') + t \in I'_{c'}$: $\nu(g(c')) + t \in I_{g(c')} \subseteq I'_{c'}$.
	\item $\nu'\circ g = (\nu\circ g + t)[R' := 0]$: $\nu' = (\nu + t)[R := 0]$, so $\nu'\circ g = (\nu + t)[R := 0]\circ g = (\nu\circ g + t)[g^{-1}(R) := 0] = (\nu\circ g + t)[R' := 0]$.
\end{itemize}
\end{proof}

\begin{lemma}
$\iota$ extends to a functor.
\end{lemma}

\begin{proof}
Given a morphism of TS between two trees $\map{f}{T = (S,i,\Gamma)}{T' = (S',i',\Gamma')}$, we have to defined a morphism of TTS $\map{\iota(f) = (f,g)}{\iota(T) = (S,i,C,\Delta)}{\iota(T') = (S',i',C',\Delta')}$, with $\map{g}{2^{\Gamma'}}{2^\Gamma}$ maps $U'$ to $\{(s,(a,t), s') \in \Gamma \mid (f(s), (a,t), f(s')) \in U'\}$. We have to prove that $\iota(f)$ is a morphism of TTS. Given $(s, (a,t), s') \in \Gamma$, let $(s, a, R, \prod\limits_{U \subseteq \Gamma} \{t_U\}, s')$ be the associated transition in $\Delta$. Since $f$ is a morphism of transition systems, then $(f(s), (a,t), f(s')) \in \Gamma'$. Let $(f(s), a, R', \prod\limits_{U' \subseteq \Gamma'} \{t'_{U'}\}, f(s'))$ be the associated transition in $\Delta'$. We must prove that:
\begin{itemize}
	\item $R' = g^{-1}(R)$: 
\begin{center}
\begin{tabular}{rclcl}
    $g^{-1}(R)$ & $=$ & $g^{-1}(\{U \subseteq \Gamma \mid (s,(a,t),s') \in U\})$ & ~~~~~ & (definition of $R$)\\
    & $=$ & $\{U' \subseteq \Gamma' \mid (s,(a,t),s') \in g(U')\}$ & ~~~~~ & \\
    & $=$ & $\{U' \subseteq \Gamma' \mid (f(s), (a,t), f(s')) \in U'\}$ & ~~~~~ & (definition of $g$)\\
    & $=$ & $R'$ & ~~~~~ & (definition of $R'$)
\end{tabular}
\end{center}
	\item for every $U' \subseteq \Gamma'$, $t_{U'} = t_{g(U')}$: if $i = q_0 \xrightarrow{(a_1,t_1)} \ldots \xrightarrow{(a_n,t_n)} q_n = s$ is the unique run from the initial state to $s$ in $T$, then if $i' = f(q_0) \xrightarrow{(a_1,t_1)} \ldots \xrightarrow{(a_n,t_n)} f(q_n) = s$ is the unique run from the initial state to $f(s)$ in $T'$, since $f$ is a morphism. So $t_{U'} = t_{g(U')}$ is a consequence of the fact that $(q_{j-1}, (a_j,t_j), q_j) \in g(U')$ iff $(f(q_{j-1}), (a_j, t_j), f(q_j)) \in U'$ by definition of $g$.
\end{itemize}
\end{proof}

Given a TTS $T = (S,i,C,\Delta)$, $\iota\circ GT$ is of the form $(S', (i,\tilde(0)), 2^{\Delta'}, \Gamma)$, where:
\begin{itemize}
	\item $S'$ is the set of runs of $\Theta T$,
	\item $\Delta'$ is the set of transitions of $G T$.
\end{itemize}
The counit $\epsilon_T$ is defined as $(f_T, g_T)$ with:
\begin{itemize}
	\item $f_T$ maps a run of $\Theta T$, which is of the form $$(q_0, \nu_0) \xrightarrow{(a_1,t_1)} \ldots \xrightarrow{(a_n,t_n)} (q_n, \nu_n)$$
	to the state $q_n$.
	\item $g_T$ maps a clock $c$ of $C$ to the set of transitions 
	$$((i,\tilde{0}) = (s_0, \nu_0) \xrightarrow{(a_1,t_1)} \ldots \xrightarrow{(a_n,t_n)} (s_n,\nu_n), (a_{n+1}, t_{n+1}), ~~~~~~~~~~$$
$$~~~~~~~~~~(s_0, \nu_0) \xrightarrow{(a_1,t_1)} \ldots \xrightarrow{(a_n,t_n)} (s_n,\nu_n) \xrightarrow{a_{n+1}, t_{n+1}} (s_{n+1}, \nu_{n+1}))$$
of $GT$ with $\nu_{n+1}(c) = 0$.
\end{itemize}

\begin{lemma}
$\epsilon_T$ is a morphism of TTS.
\end{lemma}

\begin{proof}
Assume given a transition $\delta$ of $\iota\circ GT$. Then $\delta_1 = (\pi, a_{n+1}, R, \prod\limits_{U \subseteq \Delta'} \{t_U\}, \pi')$ is defined from a transition $\delta_2$:
$$((i,\tilde{0}) = (s_0, \nu_0) \xrightarrow{(a_1,t_1)} \ldots \xrightarrow{(a_n,t_n)} (s_n,\nu_n), (a_{n+1}, t_{n+1}), ~~~~~~~~~~$$
$$~~~~~~~~~~(s_0, \nu_0) \xrightarrow{(a_1,t_1)} \ldots \xrightarrow{(a_n,t_n)} (s_n,\nu_n) \xrightarrow{a_{n+1}, t_{n+1}} (s_{n+1}, \nu_{n+1}))$$
of $G T$. This means that $((s_n,\nu_n), (a_{n+1}, t_{n+1}), (s_{n+1}, \nu_{n+1}))$ is a transition in $\Theta T$. This transition comes from a transition $(s_n, a_{n+1}, R', \prod\limits_{c \in C} I_c,s_{n+1}) \in \Delta$. Let us prove that this is the transition we are looking for:
\begin{itemize}
	\item $R' = g_T^{-1}(R)$: by construction, $R$ is the set of subsets of transitions of $GT$ that contains $\delta_2$. So $g_T^{-1}(R)$ is the set of clocks $c$ of $T$ such that $\nu_{n+1}(c) = 0$. Furthermore, since $t_{n+1} > 0$, and by definition of $R'$, $R'$ is the set of clocks $c$ of $T$ such that $\nu_{n+1}(c) = 0$.
	\item for every $c \in C$, $t_{g_T(c)} \in I_c$: by definition, $t_{g_T(c)} = \sum\limits_{j = i_c+1}^{n+1} t_j$ with $t_c = \max\{j \leq n \mid \nu_j(c) = 0\}$. So then, by induction, for every $n \geq k \geq i_c$, $\nu_k(c) = \sum\limits_{j = i_c+1}^{k} t_j$. In particular, $\nu_n(c) = \sum\limits_{j = i_c+1}^{n} t_j$, and then, by construction, $t_{g_T(c)} = \sum\limits_{j = i_c+1}^{n+1} t_j = \nu_n(c) + t_{n+1} \in I_c$.
\end{itemize}
\end{proof}

\begin{lemma}
$\epsilon$ is a natural transformation.
\end{lemma}

\begin{proof}
Assume given a morphism of TTS $(f,g)$, from $T = (S,i,C,\Delta)$ to $T' = (S',i',C',\Delta')$. We denote by $(f',g')$ the morphism $\iota\circ G(f,g)$. $f'$ maps a run of GT $$(s_0, \nu_0) \xrightarrow{(a_1,t_1)} \ldots \xrightarrow{(a_n,t_n)} (s_n,\nu_n)$$ to the run of $GT'$ $$(f(s_0), \nu_0\circ g) \xrightarrow{(a_1,t_1)} \ldots \xrightarrow{(a_n,t_n)} (f(s_n),\nu_n\circ g).$$ 
So both functions $f_{T'}\circ f'$ and $f\circ f_T$ map a run of $GT$ $$(s_0, \nu_0) \xrightarrow{(a_1,t_1)} \ldots \xrightarrow{(a_n,t_n)} (s_n,\nu_n)$$ to $f(s_n)$. Then $f_{T'}\circ f' = f\circ f_T$.

$g'$ maps a set $U$ of transitions of $GT'$ to the set of transitions of $GT$ whose image by $f'$ belongs to $U$. Then both functions $g'\circ g_{T'}$ and $g_T\circ g$ map a clock $c'$ of $C'$ to the set of transitions of $GT$ of the form:
$$((i,\tilde{0}) = (s_0, \nu_0) \xrightarrow{(a_1,t_1)} \ldots \xrightarrow{(a_n,t_n)} (s_n,\nu_n), (a_{n+1}, t_{n+1}), ~~~~~~~~~~$$
$$~~~~~~~~~~(s_0, \nu_0) \xrightarrow{(a_1,t_1)} \ldots \xrightarrow{(a_n,t_n)} (s_n,\nu_n) \xrightarrow{a_{n+1}, t_{n+1}} (s_{n+1}, \nu_{n+1}))$$
with $\nu_{n+1}(g(c')) = 0$. So $g'\circ g_{T'} = g_T\circ g$.
\end{proof}

Given a synchronization tree $T = (S,i,\Delta)$, we want to construct a morphism $\map{\eta_T}{T}{G\circ\iota T}$. Given a state $s \in S$, let $i = s_0 \xrightarrow{(a_1,t_1)} \ldots \xrightarrow{(a_n,t_n)} s$ be the unique path to $s$ in $T$. This means there are transitions $(s_{j-1}, a_j, R_j, \prod\limits_{U \subseteq \Delta} \{t_U^j\}, s_j)$ in $\iota T$. Moreover, by construction of $\iota T$, we know that:
\begin{itemize}
	\item $R_j$ is the set of subset of transitions of $T$ that contains $(s_{j-1}, (a_j, t_j), s_j)$.
	\item If we fix $t_U^0 = 0$ then $t_U^{j+1}$ is
		\begin{itemize}
			\item $t_{j+1}$ if $(s_{j-2}, (a_{j-1}, t_{j-1}), s_{j-1}) \in U$,
			\item $t_U^j + t_{j+1}$ otherwise.
		\end{itemize}
\end{itemize}
Define $\nu_j(U)$ as being $0$ if $(s_{j-1}, (a_j, t_j), s_j) \in U$, and $t_U^j$ otherwise. Then 
$$(i, \tilde{0}) = (s_0, \nu_0) \xrightarrow{(a_1,t_1)} \ldots \xrightarrow{(a_n,t_n)} (s_n, \nu_n) = (s, \nu_n)$$
is a run of $\Theta\circ\iota T$, so a state of $G\circ\iota T$. Define $\eta_T(s)$ as this run.

\begin{lemma}
$\eta_T$ is a morphism of TS.
\end{lemma}

\begin{proof}
Given $(s, (a,t),s')$ a transition of $T$, let us assume the notations form above. Define $\nu_{n+1}(U)$ as being $0$ if $(s, (a,t),s') \in U$, $\nu_n(U) + t$ otherwise. Then 
$$((s_0, \nu_0) \xrightarrow{(a_1,t_1)} \ldots \xrightarrow{(a_n,t_n)} (s_n,\nu_n), (a, t), ~~~~~~~~~~~~~$$
$$~~~~~~~~~~~~~(s_0, \nu_0) \xrightarrow{(a_1,t_1)} \ldots \xrightarrow{(a_n,t_n)} (s_n,\nu_n) \xrightarrow{a, t} (s', \nu_{n+1}))$$
is a transition of $G\circ\iota T$.
\end{proof}

We want to prove that $\eta_T$ is an isomorphism. To this end, we construct its inverse $\map{\rho_T}{G\circ\iota T}{T}$ as follows. Given a state $(i, \tilde{0}) = (s_0, \nu_0) \xrightarrow{(a_1,t_1)} \ldots \xrightarrow{(a_n,t_n)} (s_n, \nu_n)$ of $G\circ\iota T$, that is, a run of $\Theta \circ\iota T$. $\rho_T$ maps this run to $s_n$.

\begin{lemma}
$\rho_T$ is a morphism of TS.
\end{lemma}

\begin{proof}
Given a transition 
$$((i,\tilde{0}) = (s_0, \nu_0) \xrightarrow{(a_1,t_1)} \ldots \xrightarrow{(a_n,t_n)} (s_n,\nu_n), (a_{n+1}, t_{n+1}), ~~~~~~~~~~$$
$$~~~~~~~~~~(s_0, \nu_0) \xrightarrow{(a_1,t_1)} \ldots \xrightarrow{(a_n,t_n)} (s_n,\nu_n) \xrightarrow{a_{n+1}, t_{n+1}} (s_{n+1}, \nu_{n+1}))$$
of $GT$, we know that $((s_n,\nu_n), (a_{n+1}, t_{n+1}), (s_{n+1}, \nu_{n+1}))$ is a transition of $\Theta T$. So it comes from a transition $(s_n, a_{n+1}, R, \prod\limits_{U \subseteq \Delta} \{t_U\}, s_{n+1})$ in $\iota T$. By definition, $t_{\Delta} = t_{n+1}$, since $\nu_n(\Delta) = 0$. Furthermore, the transition $(s_n, a_{n+1}, R, \prod\limits_{U \subseteq \Delta} \{t_U\}, s_{n+1})$ comes from a transition $(s_n, (a_{n+1}, t), s_{n+1})$ of $T$. By definition, $t_{\Delta} = t$. So, $(s_n, (a_{n+1}, t_{n+1}), s_{n+1})$ is a transition of $T$.
\end{proof}

\begin{lemma}
$\rho_T$ is the inverse of $\eta_T$.
\end{lemma}

\begin{proof}~
\begin{itemize}
	\item The fact that $\rho_T\circ\eta_T$ is the identity is obvious.
	\item The fact that $\eta_T\circ\rho_T$ is the identity boils down to the following fact: given a state $s \in S$, there is a unique run of the form 
	$$(i,\tilde{0}) = (s_0, \nu_0) \xrightarrow{(a_1,t_1)} \ldots \xrightarrow{(a_n,t_n)} (s_n,\nu_n) = (s, \nu_n)$$
in $\Theta T$. Let us prove this fact. We already know there is at least one from the definition of $\eta_T$. Assume given another one
$$(i,\tilde{0}) = (s'_0, \nu'_0) \xrightarrow{(a'_1,t'_1)} \ldots \xrightarrow{(a'_{n'},t'_{n'})} (s'_{n'},\nu'_{n'}) = (s, \nu'_{n'})$$
Since $\rho_T$ is a morphism, 
$$i = s'_0 \xrightarrow{(a'_1,t'_1)} \ldots \xrightarrow{(a'_{n'},t'_{n'})} s'_{n'} = s$$
is a run in $T$. Since $T$ is tree, this run is the unique run 
$$i = s_0 \xrightarrow{(a_1,t_1)} \ldots \xrightarrow{(a_n,t_n)} s_n = s$$
used in the definition of $\eta_T$. Then the run 
$$(i,\tilde{0}) = (s'_0, \nu'_0) \xrightarrow{(a'_1,t'_1)} \ldots \xrightarrow{(a'_{n'},t'_{n'})} (s'_{n'},\nu'_{n'}) = (s, \nu'_{n'})$$
can only be $\eta_T(s)$, by induction on $n$.
\end{itemize}
\end{proof}

\begin{lemma}
$\rho$ (and so $\eta$) is a natural transformation.
\end{lemma}

\begin{proof}
Assume given a morphism $f$ of TS, from $T = (S,i, \Delta)$ to $T' = (S', i', \Delta')$. $G\circ\iota f$ maps a run 
$$(i,\tilde{0}) = (s_0, \nu_0) \xrightarrow{(a_1,t_1)} \ldots \xrightarrow{(a_n,t_n)} (s_n,\nu_n)$$
of $\Theta T$ to the run 
$$(i',\tilde{0}) = (f(s_0), \nu'_0) \xrightarrow{(a_1,t_1)} \ldots \xrightarrow{(a_n,t_n)} (f(s_n),\nu'_n)$$
of $\Theta T'$, where $\nu'_j(U) = \nu_j(\{(s, (a,t), s') \mid (f(s), (a, t), f(s')) \in U\})$.
So then, both functions $\rho_{T'}\circ (G\circ\iota f)$ and $f\circ\rho_T$ map a run 
$$(i,\tilde{0}) = (s_0, \nu_0) \xrightarrow{(a_1,t_1)} \ldots \xrightarrow{(a_n,t_n)} (s_n,\nu_n)$$
of $\Theta T$ to $f(s_n)$.
\end{proof}

\begin{lemma}~
\begin{itemize}
	\item For every TS $T$, $\epsilon_{\iota T} = \iota(\rho_T)$.
	\item For every TTS $T$, $G(\epsilon_T) = \rho_{GT}$.
\end{itemize}
\end{lemma}

\begin{proof}~
\begin{itemize}
	\item Fix $\epsilon_{\iota T} = (f_1, g_1)$ and $\iota(\rho_T) = (f_2, g_2)$. Both morphisms are from $\iota\circ G\circ\iota T$ to $\iota T$. The states of $\iota\circ G\circ\iota T$ are runs of $\Theta\circ\iota T$, that is runs of the form 
$$(i,\tilde{0}) = (s_0, \nu_0) \xrightarrow{(a_1,t_1)} \ldots \xrightarrow{(a_n,t_n)} (s_n,\nu_n)$$
But as we have seen earlier, this means that 
$$i = s_0 \xrightarrow{(a_1,t_1)} \ldots \xrightarrow{(a_n,t_n)} s_n$$
is a run of $T$ and the $\nu_j$ are uniquely determined by this run.

Both $f_1$ and $f_2$ map such a run 
$$(i,\tilde{0}) = (s_0, \nu_0) \xrightarrow{(a_1,t_1)} \ldots \xrightarrow{(a_n,t_n)} (s_n,\nu_n)$$
of $\Theta\circ\iota T$ to $s_n$.

Both $g_1$ and $g_2$ map a set $U$ of transition of $T$ (which is a clock of $\iota T$), to the set of transitions 
$$((i,\tilde{0}) = (s_0, \nu_0) \xrightarrow{(a_1,t_1)} \ldots \xrightarrow{(a_n,t_n)} (s_n,\nu_n), (a_{n+1}, t_{n+1}), ~~~~~~~~~~$$
$$~~~~~~~~~~(s_0, \nu_0) \xrightarrow{(a_1,t_1)} \ldots \xrightarrow{(a_n,t_n)} (s_n,\nu_n) \xrightarrow{a_{n+1}, t_{n+1}} (s_{n+1}, \nu_{n+1}))$$
of $G\circ\iota T$ such that $(s_n, (a_{n+1}, t_{n+1}), s_{n+1}) \in U$.
	\item Both morphisms are from $G\circ\iota\circ GT$ to $GT$. A state of $G\circ\iota\circ GT$ is a run 
$$(\pi_0, \nu_0) \xrightarrow{(a_1,t_1)} \ldots \xrightarrow{(a_n,t_n)} (\pi_n,\nu_n)$$
of $\Theta\circ\iota\circ GT$, which means that $\pi_j$ are themselves runs of $\Theta T$, of the form:
\begin{itemize}
	\item $\pi_0$ is the singleton $(i, \tilde{0})$,
	\item $\pi_{j+1}$ is obtained from $\pi_j$ by extending it with a $(a_{j+1}, t_{j+1})$ transition.
\end{itemize}
This means there are $s_j$, $\nu'_j$ such that 
$$\pi_j = (s_0, \nu'_0) \xrightarrow{(a_1,t_1)} \ldots \xrightarrow{(a_j,t_j)} (s_j, \nu'_j)$$
Then $\rho_{GT}$ maps 
$$(\pi_0, \nu_0) \xrightarrow{(a_1,t_1)} \ldots \xrightarrow{(a_n,t_n)} (\pi_n,\nu_n)$$
to $\pi_n$, while $G(\epsilon_T)$ maps it to
$$(s_0, \nu_0\circ g_T) \xrightarrow{(a_1,t_1)} \ldots \xrightarrow{(a_n,t_n)} (s_n, \nu_n\circ g_T)$$
To conclude, it is then enough to prove that $\nu'_j = \nu_j\circ g_T$, by induction on $j$:
\begin{itemize}
	\item $\nu'_0 = \tilde{0} = \nu_0\circ g_T$.
	\item By construction, the transition $(\pi_j, (a_{j+1}, t_{j+1}), \pi_{j+1})$ of $GT$ induces a transition of the form $\delta = (\pi_j, a_{j+1}, R_j, \prod\limits_{U \subseteq \Gamma} t_U^j, \pi_{j+1})$ of $\iota\circ GT$. It is also induced by a transition $\delta' = (s_j, a_{j+1}, R'_j, \prod\limits_{c \in C} I_c, s_{j+1})$. Furthermore, since $GT$ is a tree, this transition is the unique one of the form $(\pi_j, a_{j+1}, \_, \_, \pi_{j+1})$. This means that the transition $((\pi_j, \nu_j), (a_j, t_j), (\pi_{j+1}, \nu_{j+1}))$ of $\Theta\circ\iota\circ GT$ is induced by $\delta$. In total, this means that:
	\begin{itemize}
		\item for every clock $c \in C$, $\nu'_{j+1}(c) = 0$ if $c \in R'_j$ or $\nu_j(c) + t_{j+1}$.
		\item for every $U \subseteq \Gamma$, $\nu_{j+1}(U) = 0$ if $U \in R_j$ or $\nu_j(U) + t_{j+1}$.
	\end{itemize}
To conclude, it is enough to prove that for every clock $c \in C$, $c \in R'_j$ iff $g_T(c) \in R_j$. 
\begin{center}
\begin{tabular}{rcl}
    $c \in R'_j$ & $\Leftrightarrow$ & $\nu'_{j+1}(c) = 0$\\
     & $\Leftrightarrow$ & $(\pi_j, (a_{j+1}, t_{j+1}), \pi_{j+1}) \in g_T(c)$\\
     & $\Leftrightarrow$ & $g_T(c) \in R'_j$
\end{tabular}
\end{center}
\end{itemize}
	
\end{itemize}
\end{proof}

\section{Omitted proofs on hybrid systems}

\begin{lemma}
\label{lem:moves}
Let $f = (f_M, f_I)$ be a morphism from $T$ to $T'$. If $T$ moves from $(m,\sigma)$ to $(m',\sigma')$ by doing the action $a$, with time $t$, then $T'$ moves from $(f_M(m),f_X(\sigma))$ to $(f_M(m'),f_X(\sigma'))$ by doing the action $a$, with time $t$. Consequently, if $$(m_0, \sigma_0) \xrightarrow{a_1, t_1} \ldots \xrightarrow{a_k, t_k} (m_k, \sigma_k)$$ is a run of $T$, $$(f_M(m_0), f_X(\sigma_0)) \xrightarrow{a_1, t_1} \ldots \xrightarrow{a_k, t_k} (f_M(m_k), f_X(\sigma_k))$$ is a run of $T'$.
\end{lemma}

\begin{proof}
If $T$ moves from $(m,\sigma)$ to $(m',\sigma')$ by doing the action $a$, with time $t$, then this means that:
\begin{itemize}
	\item $e = (m, a, m') \in E$,
	\item for every $i \in I$, there is a derivable function $\map{x_i}{[0,t]}{\mathbb{R}^{n_i}}$ such that:
		\begin{itemize}
				\item for every $s \in [0,t]$, $\dot{x}(s) = F_{m,i}(s,x_i(s))$,
				\item for every $s \in [0,t]$, $x(s) \in I_{m}$,
				\item $(x_i(0))_{i\in I} = \sigma$,
				\item $(x_i(t))_{i\in I} \in G_e$,
				\item $\sigma' = (R_{e,i}(x_i(t)))_{i\in I}$.
		\end{itemize}
\end{itemize}
Now, for every $i' \in I'$, consider $\map{y_{i'}}{[0,t]}{\mathbb{R}^{n_{i'}}}$, defined as $y_{i'} = x_{f_I(i')}$. This function is derivable. Furthermore, $f(e) = (f_M(m), a, f_M(m)) \in E'$ by definition of a morphism and:
			\begin{itemize}
				\item for every $s \in [0,t]$, 
\begin{center}
\begin{tabular}{rcll}
    $\dot{y_{i'}}(s)$ & $=$ & $\dot{x_{f_I(i')}}(s)$ & (definition of $y_i$)\\
    & $=$ & $F_{m,f_I(i')}(s,x_{f_I(i')}(s))$ & (hypothesis on $x_{f_I(i')}$)\\
    & $=$ & $F'_{f_M(m), i'}(s, x_{f_I(i')}(s))$ & ($f$ is a morphism)\\
    & $=$ & $F'_{f_M(m), i'}(s, y_{i'}(s))$ & (definition of $y_{i'}$)
\end{tabular}
\end{center}
				\item for every $s \in [0,t]$, $(y_{i'}(s))_{i' \in I'} = f_{X}((x_{i}(s))_{i \in I}) \in I_{f_M(m)}$, 
				\item $(y_{i'}(0))_{i'\in I'} = (x_{f_I(i')}(0))_{i'\in I'} = f_X((x_{i}(0))_{i \in I}) = f_X(\sigma)$,
				\item $(y_{i'}(t))_{i'\in I'} = f_X((x_i(t))_{i\in I}) \in G_{f(e)}$,
				\item $f_X(\sigma') = (R_{e,f_I(i')}(x_{f_I(i')}(t)))_{i' \in I'} = (R'_{f(e), i'}(x_{f_I(i')}(t)))_{i'\in I'} = (R'_{f(e), i'}(y_{i'}))_{i' \in I'}$.
			\end{itemize}
		This witnesses the fact that $T'$ moves from $(f_M(m),f_X(\sigma))$ to $(f_M(m'),f_X(\sigma'))$ by doing the action $a$, with time $t$.
\end{proof}

\begin{lemma}
$\hysy$ is a category, whose composition is given by:
$$(g_M, g_I)\circ(f_M, f_I) = (g_M\circ f_M, f_I\circ g_I)$$
\end{lemma}

\begin{proof}
Assume that $f = (f_M, f_I)$ from $T$ to $T'$ is $\epsilon$-bounded and $g = (g_M, g_I)$ from $T'$ to $T''$ is $\epsilon'$-bounded, and let us prove that $g\circ f$ is $\epsilon+\epsilon'$-bounded. 
\begin{itemize}
	\item First, let us type check $g\circ f$:
	\begin{itemize}
		\item $\map{f_M}{M}{M'}$ and $\map{g_M}{M'}{M''}$, so $\map{g_M\circ f_M}{M}{M''}$.
		\item $\map{f_I}{I'}{I}$ and $\map{g_I}{I''}{I'}$, so $\map{f_I\circ g_I}{I''}{I}$.
	\end{itemize}
	\item We want to prove the requirement for $g\circ f$ to be an $\epsilon+\epsilon'$-bounded morphism:
		\begin{itemize}
			\item $g_M\circ f_M(m_0) = g_M(m'_0) = m''_0$.
			\item $n_{i''} = n_{g_I(i'')} = n_{f_I(g_I(i''))}$.
			\item Let us denote by $(g\circ f)_X$, the function from $\prod\limits_{i\in I} \mathbb{R}^{n_i}$ to $\prod\limits_{i''\in I''} \mathbb{R}^{n_{i''}}$ defined by:
	$$(x_i)_{i\in I} \mapsto (x_{f_I\circ g_I(i'')})_{i'' \in I''}$$
Then $(g\circ f)_X = g_X \circ f_X$
			\item $(g\circ f)_X(\sigma_0) = g_X\circ f_X(\sigma_0) = g_X(\sigma'_0) = \sigma''_0$.
			\item if $(m, a, m') \in E$, then $(f_M(m), a, f_M(m')) \in E'$, and so $(g_M\circ f_M(m), a, g_M\circ f_M(m')) \in E''$.
			\item if $\sigma \in I_m$, then $f_X(\sigma) \in I'_{f_M(m)}$, and so $(g\circ f)_X(\sigma) = g_X\circ f_X(\sigma) \in I''_{g_M\circ f_M(m)}$.
			\item idem, if $\sigma \in G_e$, $(g\circ f)_X(\sigma) \in G''_{g\circ f(e)}$.
			\item let $x_{i''} \in \mathbb{R}^{n_{f_I\circ g_I(i'')}}$ and $t \in \mathbb{R}$, 
\begin{center}
\begin{tabular}{rcll}
    $F_{m,f_I\circ g_I(i'')}(t, x_{i''})$ & $=$ & $F'_{f_M(m), g_I(i'')}(t, x_{i''})$ & ($f$ morphism)\\
    & $=$ & $F''_{g_M\circ f_M(m), i''}(t, x_{i''})$ & ($g$ morphism)
\end{tabular}
\end{center}			
			\item $R_{e, f_I\circ g_I(i'')} = R'_{f(e), g_I(i'')}= R''_{g\circ f(e), i''}$.
			\item let 
			$$(m_0, \sigma_0) \xrightarrow{a_1, t_1} \ldots \xrightarrow{a_k, t_k} (m_k, \sigma_k)$$
			be a run of $T$, then 
			$$(f_M(m_0), f_X(\sigma_0)) \xrightarrow{a_1, t_1} \ldots \xrightarrow{a_k, t_k} (f_M(m_k), f_X(\sigma_k))$$
			is a run of $T'$. We would like to prove that $$d(o(m_k,\sigma_k), o''(g_M\circ f_M(m_k),(g\circ f)_X(\sigma_k))) \leq \epsilon + \epsilon'$$
\begin{center}
\begin{tabular}{rcll}
    $d(o(m_k,\sigma_k), o''(g_M\circ f_M(m_k),(g\circ f)_X(\sigma_k)))$ & $=$ & $d(o(m_k,\sigma_k), o''(g_M\circ f_M(m_k),g_X\circ f_X(\sigma_k)))$ & (cf. previous point)\\
    & $\leq$ & $d(o(m_k,\sigma_k),o'(f_M(m_k), f_X(\sigma_k)))$  & \\
    & ~ & $+ d(o'(f_M(m_k), f_X(\sigma_k)), o''(g_M\circ f_M(m_k),g_X\circ f_X(\sigma_k)))$ & (triangular inequality)\\
    & $\leq$ & $\epsilon + \epsilon'$ & ($f$ and $g$ bounded)
\end{tabular}
\end{center}
		\end{itemize}
\end{itemize}
\end{proof}

\begin{lemma}
$H$ extends to a functor.
\end{lemma}

\begin{proof}
Assume given an $\epsilon$-bounded morphism $f = (f_M, f_I)$ from $T$ to $T'$. $H f$ will map a run 
$$(m_0, \sigma_0) \xrightarrow{a_1, t_1} \ldots \xrightarrow{a_k, t_k} (m_k, \sigma_k)$$ of $T$ to the run $$(f_M(m_0), f_X(\sigma_0)) \xrightarrow{a_1, t_1} \ldots \xrightarrow{a_k, t_k} (f_M(m_k), f_X(\sigma_k))$$ of $T'$. Let us prove that $H f$ is a well-defined $\epsilon$-bounded morphism of transition systems with observations.
\begin{itemize}
	\item First, $H f$ is well-defined, by lemma \ref{lem:moves}.
	\item $Hf(m_0, \sigma_0) = (f_M(m_0), f_X(\sigma_0)) = (m'_0, \sigma'_0)$.
	\item The fact that $Hf$ preserves the extensions of runs, that is, maps transitions of $HT$ to transitions of $HT'$, is again a consequence of lemma \ref{lem:moves}.
	\item Let $\pi$ be a run of $T$ of the form:
	$$(m_0, \sigma_0) \xrightarrow{a_1, t_1} \ldots \xrightarrow{a_k, t_k} (m_k, \sigma_k)$$
	and denote by $f\pi$ the run
	$$(f_M(m_0), f_X(\sigma_0)) \xrightarrow{a_1, t_1} \ldots \xrightarrow{a_k, t_k} (f_M(m_k), f_X(\sigma_k))$$
	 of $T'$. Then $d(\omega(\pi), \omega'(f\pi') = d(o(m_k,\sigma_k), o(f_M(m_k), f_X(\sigma_k))) \leq \epsilon$.
\end{itemize}
\end{proof}

\begin{lemma}
$\iota$ extends to a functor from $\trorea$ to $\hysy$.
\end{lemma}

\begin{proof}
Given an $\epsilon$-bounded morphism of synchronization trees with observations from $T = (S, i, \Delta, \omega)$ to $T' = (S', i', \Delta', \omega')$, define $\iota f = (f_M, f_I)$ where:
\begin{itemize}
	\item $\map{f_M}{S}{S'}$ is given by $f$.
	\item $\map{f_I}{I'}{I}$ maps $(n, \sigma, (F_{s'})_{s' \in S}, (R_{(s',a,t')})_{(s',a,t') \in E'})$ to $(n, \sigma, (F_{f(s)})_{s \in S}, (R_{(f(s), a, f(t))})_{(s,a,t) \in E})$.
\end{itemize}
Let us prove it is an $\epsilon$-bounded morphism of hybrid systems from $\iota T$ to $\iota T'$:
\begin{itemize}
	\item $f_M(i) = f(i) = i'$.
	\item $n_{(n, \sigma, (F_{s'})_{s' \in S}, (R_{(s',a,t')})_{(s',a,t') \in E'})} = n = n_{(n, \sigma, (F_{f(s)})_{s \in S}, (R_{(f(s), a, f(t))})_{(s,a,t) \in E})}$.
	\item ~
\vspace{-0.3cm}
\begin{center}
\begin{tabular}{rcll}
    $f_X(\sigma_0)$ & $=$ & $f_X((\sigma)_{(n, \sigma, (F_s)_{s \in S}, (R_e)_{e \in E})})$ & (definition of $\sigma_0$)\\
    & $=$ & $(\sigma)_{(n, \sigma, (F_{s'})_{s' \in S'}, (R_{e'})_{e' \in E'})}$ & (definition of $f_X$)\\
    & $=$ & $\sigma'_0$ & (definition of $\sigma'_0$)
\end{tabular}
\end{center}
	\item if $(s, a, s') \in E$, that is, if there is $t \in \rea$, such that $(s, (a,t), s') \in \Delta$, then, since $f$ is a morphism, $(f(s), (a,t), f(s')) \in \Delta'$, and so $(f(s), a, f(s')) \in E'$.
	\item assume $\sigma \in G_{(s,a,s')}$, that is, $\sigma = (x_\alpha^{k+1}(t))_{\alpha\in I}$ with:
		\begin{itemize}
			\item $t$ is the unique time such that $(s,(a,t),s') \in \Delta$,
			\item $i = s_0 \xrightarrow{a_1, t_1} \ldots \xrightarrow{a_k, t_k} s_k = s$ is the unique path from $i$ to $s$, 			\item $a_{k+1} = a$, $t_{k+1} = t$, $s_{k+1} = s'$, 
			\item for any $\alpha \in I$, so of the form $\alpha = (n, \sigma, (F_s)_{s\in S}, (R_e)_{e\in E})$, and every $1 \leq j \leq k+1$, $\map{x_\alpha^j}{[0,t_j]}{\mathbb{R}^n}$ is the solution of $\dot{x}(u) = F_{s_{j-1}}(u,x)$, with the following initial condition:
				\begin{itemize}
					\item $x(0) = \sigma$ if $j = 1$,
					\item $x(0) = R_{(s_{j-2}, a_{j-1}, s_{j-1})}(x_\alpha^{j-1}(t_{j-1}))$ if $j > 1$.
				\end{itemize}
			\end{itemize}
		Let us prove that $f_X(\sigma) \in G'_{(f(s), a, f(s'))}$. For that, observe that:
		\begin{itemize}
			\item $f_X(\sigma) = (x_{f_I(\beta)}^{k+1}(t))_{\beta\in I'}$,
			\item $t$ is the unique time such that $(f(s),(a,t),f(s')) \in \Delta'$,
			\item $i' = f(s_0) \xrightarrow{a_1, t_1} \ldots \xrightarrow{a_k, t_k} f(s_k) = f(s)$ is the unique path from $i'$ to $f(s)$,
			\item for any $\beta \in I'$, so of the form $\beta = (n, \sigma, (F_{s'})_{s'\in S'}, (R_{e'})_{e'\in E'})$, and every $1 \leq j \leq k+1$, $\map{x_{f_I(\beta)}^j}{[0,t_j]}{\mathbb{R}^n}$ is the solution of $\dot{y}(u) = F_{f(s_{j-1})}(u,y)$, with the following initial condition:
				\begin{itemize}
					\item $y(0) = f_X(\sigma)$ if $j = 1$,
					\item $y(0) = R_{(f(s_{j-2}), a_{j-1}, f(s_{j-1}))}(x_{f_I(\beta)}^{j-1}(t_{j-1}))$ if $j > 1$.
				\end{itemize}
		\end{itemize}
	\item similarly, if $\sigma \in I_s$, $f_X(\sigma) \in I'_{f(s)}$.
	\item the next two conditions are obvious by definition of $f_I$.
	\item given a run 
	$$(s_0, \sigma_0) \xrightarrow{a_1, t_1} \ldots \xrightarrow{a_k, t_k} (s_k, \sigma_k)$$
	of $\iota T$, $$d(o(s_k, \sigma_k), o'(f(s_k), f_X(\sigma_k))) = d(\omega(s_k), \omega'(f(s_k))) \leq \epsilon.$$
\end{itemize}
\end{proof}

Given an hybrid system $T = (M, I, (n_i)_{i\in I}, E, (G_e)_{e\in E}, (R_{e,i})_{e\in E, i\in I}, (F_{m,i})_{m\in M, i\in I}, (I_m)_{m\in M}, m_0, \sigma_0, o)$, $\iota\circ HT$ is of the form:
$$(M', I', (n'_{i'})_{i'\in I'}, E', (G'_{e'})_{e'\in E'}, (R'_{e',i'})_{e'\in E', i'\in I'}, (F'_{m',i'})_{m'\in M', i'\in I'}, (I'_{m'})_{m'\in M'}, m'_0, \sigma'_0, o')$$
with:
\begin{itemize}
	\item $M'$ is the set of runs of $T$,.
	\item $E'$ is the set of transition of $HT$ for which we have forgotten the time, that is, triples:
	$$((m_0,\sigma_0) \xrightarrow{a_1, t_1} \ldots \xrightarrow{a_n,t_n} (m_n, \sigma_n), a_{n+1},~~~~~~~~~~~~~~~~~~~~$$
	$$~~~~~~~~~~~~~~~~~~~~(m_0,\sigma_0) \xrightarrow{a_1, t_1} \ldots \xrightarrow{a_n,t_n} (m_n, \sigma_n) \xrightarrow{a_{n+1}, t_{n+1}} (m_{n+1}, \sigma_{n+1}))$$
	where $$(m_0,\sigma_0) \xrightarrow{a_1, t_1} \ldots \xrightarrow{a_n,t_n} (m_n, \sigma_n)$$ and $$(m_0,\sigma_0) \xrightarrow{a_1, t_1} \ldots \xrightarrow{a_n,t_n} (m_n, \sigma_n) \xrightarrow{a_{n+1}, t_{n+1}} (m_{n+1}, \sigma_{n+1})$$ are runs of $T$.
	\item $I'$ is the set of quadruples $(n, \sigma, (F_{m'})_{m' \in M'}, (R_{e'})_{e' \in E'})$ where:
		\begin{itemize}
			\item $n$ is am integer,
			\item $\sigma \in \mathbb{R}^n$,
			\item $\map{F_{m'}}{\mathbb{R}\times\mathbb{R}^n}{\mathbb{R}^n}$ is a function which is continuous and locally Lipschitz on the second argument,
			\item $\map{R_{e'}}{\mathbb{R}^n}{\mathbb{R}^n}$ is a function.
		\end{itemize}
	\item $n'_{(n, \sigma, (F_{m'})_{m' \in M'}, (R_{e'})_{e' \in E'})} = n$.
	\item Given a run $\pi$ of $T$ of the form:
	$$(m_0,\sigma_0) \xrightarrow{a_1, t_1} \ldots \xrightarrow{a_k,t_k} (m_k, \sigma_k)$$
	write $\pi_j$ for the run:
	$$(m_0,\sigma_0) \xrightarrow{a_1, t_1} \ldots \xrightarrow{a_j,t_j} (m_j, \sigma_j).$$
	Given additionally $\alpha = (n, \sigma, (F_{m'})_{m' \in M'}, (R_{e'})_{e' \in E'}) \in I'$ and $t$, define $\sigma_{\pi, \alpha}(t) \in \mathbb{R}^n$ as $x_{\alpha}^k(t)$, where $x_{\alpha}^j$ is the solution on $[0, t_{j+1}]$ (with the convention that $t_{k+1} = t$) of the equation $\dot{x}(s) = F_{\pi_j}(s,x(s))$ with the following initial condition:
		\begin{itemize}
			\item $x(0) = \sigma$ if $j = 0$,
			\item $x(0) = R_{(\pi_{j-1}, a_{j}, \pi_{j})}(x_{\alpha}^{j-1}(t_{j}))$ if $j > 0$.
		\end{itemize}
	\item $G'_{e'}$, where $e'$ is of the form 
	$$(\pi, a_{n+1}, \pi') = ((m_0,\sigma_0) \xrightarrow{a_1, t_1} \ldots \xrightarrow{a_n,t_n} (m_n, \sigma_n), a_{n+1},~~~~~~~~~~~~~~~~~~~~$$
	$$~~~~~~~~~~~~~~~~~~~~(m_0,\sigma_0) \xrightarrow{a_1, t_1} \ldots \xrightarrow{a_n,t_n} (m_n, \sigma_n) \xrightarrow{a_{n+1}, t_{n+1}} (m_{n+1}, \sigma_{n+1}))$$
	is given by $\{(\sigma_{\pi, \alpha}(t_{n+1}))_{\alpha \in I'}\}$.
	\item $I'_{m'}$, when $m'$ is a run of the form:
	$$(m_0,\sigma_0) \xrightarrow{a_1, t_1} \ldots \xrightarrow{a_n,t_n} (m_n, \sigma_n)$$
	is given by the set $$\{(\sigma_{m', \alpha}(t))_{\alpha \in I'} \mid \exists t_{n+1}, a_{n+1}, m_{n+1}, \sigma_{n+1}, \, (m_n,\sigma_n) \xrightarrow{a_{n+1}, t_{n+1}} (m_{n+1}, \sigma_{n+1}) \wedge t \leq t_{n+1}\}.$$
	\item Given $e'' \in E'$, and $i' = (n, \sigma, (F_{m'})_{m' \in M'}, (R_{e'})_{e' \in E'})$, $R'_{e'', i'} = R_{e''}$.
	\item Given $m'' \in M'$, and $i' = (n, \sigma, (F_{m'})_{m' \in M'}, (R_{e'})_{e' \in E'})$, $F'_{m'', i'} = F_{m''}$.
	\item $m'_0$ is given by the initial run of $T$, that is, the singleton $(m_0, \sigma_0)$. 
	\item $\sigma'_0$ is given by the collection $(\sigma)_{(n, \sigma, (F_{s'})_{s' \in S'}, (R_{e'})_{e' \in E'})}$.
	\item $o'(m', \sigma)$ where $m'$ is a run of the form:
	$$(m_0,\sigma_0) \xrightarrow{a_1, t_1} \ldots \xrightarrow{a_n,t_n} (m_n, \sigma_n)$$
	is given by $o'(m_n, \sigma_n)$.
\end{itemize}

Now $\map{\epsilon_T = (f_{T, M}, f_{T, I})}{\iota\circ HT}{T}$, is given by the following:
\begin{itemize}
	\item $f_{T, M}$ maps a run $(m_0,\sigma_0) \xrightarrow{a_1, t_1} \ldots \xrightarrow{a_k,t_k} (m_k, \sigma_k)$ of $T$ to $m_k$.
	\item $f_{T,I}$ maps a subsystem $i$ of $T$ to the quadruple $\alpha_i = (n_i, \sigma_{i,0}, (F_{m'})_{m' \in M'}, (R_{e'})_{e' \in E'})$, where:
		\begin{itemize}
			\item $\sigma_0 = (\sigma_{i,0})_{i \in I}$,
			\item when $m'$ is a run of $T$ of the form 
			$$(m_0,\sigma_0) \xrightarrow{a_1, t_1} \ldots \xrightarrow{a_k,t_k} (m_k, \sigma_k)$$
			then $F_{m'} = F_{m_k, i}$,
			\item when $e'$ is a transition of $HT$ of the form
			$$((m_0,\sigma_0) \xrightarrow{a_1, t_1} \ldots \xrightarrow{a_n,t_n} (m_n, \sigma_n), a_{n+1},~~~~~~~~~~~~~~~~~~~~$$
	$$~~~~~~~~~~~~~~~~~~~~(m_0,\sigma_0) \xrightarrow{a_1, t_1} \ldots \xrightarrow{a_n,t_n} (m_n, \sigma_n) \xrightarrow{a_{n+1}, t_{n+1}} (m_{n+1}, \sigma_{n+1}))$$
			then $R_{e'} = R_{(m_n, a_{n+1}, m_{n+1}), i}$.
		\end{itemize}
\end{itemize}

\begin{lemma}
\label{lem:guainv}
Given $e' \in E'$, of the form
	$$(\pi, a_{n+1},\pi') = ((m_0,\sigma_0) \xrightarrow{a_1, t_1} \ldots \xrightarrow{a_n,t_n} (m_n, \sigma_n), a_{n+1},~~~~~~~~~~~~~~~~~~~~$$
	$$~~~~~~~~~~~~~~~~~~~~(m_0,\sigma_0) \xrightarrow{a_1, t_1} \ldots \xrightarrow{a_n,t_n} (m_n, \sigma_n) \xrightarrow{a_{n+1}, t_{n+1}} (m_{n+1}, \sigma_{n+1}))$$
if we denote by $\mu_i$ the function $\sigma_{\pi, f_{T,I}(i)}$, we have the following:
\begin{itemize}
	\item for every $t \leq t_{n+1}$, $(\mu_i(t))_{i \in I} \in I_{m_n}$.
	\item $(\mu_i(t_{n+1}))_{i \in I} \in G_{(m_n, a_{n+1}, m_{n+1})}$.
	\item $(R_{(m_n, a_{n+1}, m_{n+1}), i}(\mu_i))_{i \in I} = \sigma_{n+1}$.
\end{itemize}
\end{lemma}

\begin{proof}
This is done by induction on the $n$, by using the unicity of the solution from Picard-Lindel\"of theorem.
\end{proof}

\begin{lemma}
$\epsilon_T$ is a $0$-bounded morphism of hybrid systems.
\end{lemma}

\begin{proof}~
\begin{itemize}
	\item $f_{T,M}(m'_0) = f_{T,M}((m_0, \sigma)) = m_0$.
	\item $f_{T,X}(\sigma'_0) = (\sigma'_{f_{T,I}(i),0})_{i \in I} = (\sigma_i)_{i \in I} = \sigma_0$.
	\item Let $e' \in E'$, of the form
	$$((m_0,\sigma_0) \xrightarrow{a_1, t_1} \ldots \xrightarrow{a_n,t_n} (m_n, \sigma_n), a_{n+1},~~~~~~~~~~~~~~~~~~~~$$
	$$~~~~~~~~~~~~~~~~~~~~(m_0,\sigma_0) \xrightarrow{a_1, t_1} \ldots \xrightarrow{a_n,t_n} (m_n, \sigma_n) \xrightarrow{a_{n+1}, t_{n+1}} (m_{n+1}, \sigma_{n+1}))$$
	Then, in particular, $(m_0,\sigma_0) \xrightarrow{a_1, t_1} \ldots \xrightarrow{a_n,t_n} (m_n, \sigma_n) \xrightarrow{a_{n+1}, t_{n+1}} (m_{n+1}, \sigma_{n+1})$ is a run of $T$, which means that $T$ moves from $(m_n, \sigma_n)$ to $(m_{n+1}, \sigma_{n+1})$, by doing the $a_{n+1}$ action, with time $t_{n+1}$. In particular, $(m_n, a_{n+1}, m_{n+1}) \in E$.
	\item If $\sigma \in G_{e'}$, then $\sigma = (\sigma_{\pi, i'}(t_{n+1}))_{i' \in I'}$, with the notations from above. Then $f_{T,X}(\sigma) \in G_{(m_n, a_{n+1}, m_{n+1})}$, by lemma \ref{lem:guainv}.
	\item If $m'$ is a run of $T$ of the form
	$$\pi = (m_0,\sigma_0) \xrightarrow{a_1, t_1} \ldots \xrightarrow{a_n,t_n} (m_n, \sigma_n)$$
	$\sigma \in I_{m'}$, then $\sigma = (\sigma_{\pi, i'}(t))_{i' \in I'}$, for some $t \leq t_{n+1}$, with 
	$$\pi' = (m_0,\sigma_0) \xrightarrow{a_1, t_1} \ldots \xrightarrow{a_n,t_n} (m_n, \sigma_n) \xrightarrow{a_{n+1}, t_{n+1}} (m_{n+1}, \sigma_{n+1}))$$
	being a run of $T$. Then use lemma \ref{lem:guainv} with $e' = (\pi, a_{n+1}, \pi')$.
	\item the following two axioms are obvious by definition of $f_{T,I}$.
	\item Given a run 
	$$(m'_0,\sigma'_0) \xrightarrow{a_1, t_1} \ldots \xrightarrow{a_n,t_n} (m'_n, \sigma'_n)$$
	of $\iota HT$, then, by lemma \ref{lem:guainv}, $m'_j$ is the run of $T$ of the form:
	$$(f_{T,M}(m'_0),f_{T,X}(\sigma'_0)) \xrightarrow{a_1, t_1} \ldots \xrightarrow{a_n,t_n} (f_{T,M}(m'_j), f_{T,X}(\sigma'_j))$$
	and so $$d(o'(m'_n, \sigma'_n), o(f_{T,M}(m'_n), f_{T,X}(\sigma'_n))) = d(o(f_{T,M}(m'_n), f_{T,X}(\sigma'_n)), o(f_{T,M}(m'_n), f_{T,X}(\sigma'_n))) = 0.$$
\end{itemize}
\end{proof}

\begin{lemma}
$\epsilon$ is a natural transformation.
\end{lemma}

\begin{proof}
Assume given a bounded morphism of hybrid systems $g = (g_M, g_I)$ from $T$ to $T'$. Denote $\iota H g$ by $(h_M, h_I)$. Then:
\begin{itemize}
	\item $h_M$ maps a run 
	$$\pi = (m_0,\sigma_0) \xrightarrow{a_1, t_1} \ldots \xrightarrow{a_n,t_n} (m_n, \sigma_n)$$
	of $T$, to the run 
	$$g\pi = (g_M(m_0),g_X(\sigma_0)) \xrightarrow{a_1, t_1} \ldots \xrightarrow{a_n,t_n} (g_M(m_n), g_X(\sigma_n))$$
	of $T'$.
	\item $g_I$ maps $$(n, \sigma, (F_{\pi'})_{\pi' \text{ run of } T'}, (R_e')_{e' \text{ extension of runs of } T'})$$
	to $$(n, \sigma, (F_{g\pi}) _{\pi \text{ run of } T}, (R_{(f\pi_1, a, f\pi_2)})_{(\pi_1, a, \pi_2) \text{ extension of runs of } T})$$
\end{itemize}
Then:
\begin{itemize}
	\item $g_M\circ f_{T,M}$ and $f_{T',M}\circ h_M$ both map a run
	$$(m_0,\sigma_0) \xrightarrow{a_1, t_1} \ldots \xrightarrow{a_n,t_n} (m_n, \sigma_n)$$
	of $T$ to $g_M(m_n)$.
	\item Given a subsystem $i'$ of $T'$, $f_{T,I}\circ g_I(i')$ is
	$$(n_{g_I(i')}, \sigma_{g_I(i'),0}, (F_{\pi})_{\pi \text{ run of } T}, (R_e)_{e' \text{ extension of runs of } T})$$
	and $h_I\circ f_{T', I}(i')$ is
	$$(n_{i'}, \sigma_{i',0}, (F'_{\pi})_{\pi \text{ run of } T}, (R'_e)_{e' \text{ extension of runs of } T})$$
	with:
	\begin{itemize}
		\item given a run 
		$$\pi = (m_0,\sigma_0) \xrightarrow{a_1, t_1} \ldots \xrightarrow{a_n,t_n} (m_n, \sigma_n)$$
		of $T$, $F_\pi = F_{m_k, g_I(i')}$ and $F'_\pi = F'_{g_M(m_k), i'}$ which are both equal since $g$ is a morphism.
		\item given an extension $e$ of runs
		$$((m_0,\sigma_0) \xrightarrow{a_1, t_1} \ldots \xrightarrow{a_n,t_n} (m_n, \sigma_n), a_{n+1},~~~~~~~~~~~~~~~~~~~~$$
	$$~~~~~~~~~~~~~~~~~~~~(m_0,\sigma_0) \xrightarrow{a_1, t_1} \ldots \xrightarrow{a_n,t_n} (m_n, \sigma_n) \xrightarrow{a_{n+1}, t_{n+1}} (m_{n+1}, \sigma_{n+1}))$$
	of $T$, $R_e = R_{(m_n, a_{n+1}, m_{n+1}), g_I(i')}$ and $R'_e = R'_{(g_M(m_n), a_{n+1}, g_M(m_{n+1})), i'}$, which are both equal since $g$ is a morphism.
	\end{itemize}
	Furthermore, since $g$ is a morphism we have that $n_{g_I(i')} = n_{i'}$ and $\sigma_{g_I(i'),0} = \sigma_{i',0}$, Therefor, $f_{T,I}\circ g_I = h_I\circ f_{T', I}$.
\end{itemize}
\end{proof}

Given a synchronization tree with observations $T = (S,i,\Delta, o)$, we want to define a morphism $\map{\eta_T}{T}{H\iota T}$. 

\begin{lemma}
\label{lem:unicrun}
Given a run of 
$$s_0 \xrightarrow{a_1, t_1} \ldots \xrightarrow{a_n, t_n} s_n$$
of $T$, there is a unique run of the form
$$(s_0, \sigma_0) \xrightarrow{a_1, t_1} \ldots \xrightarrow{a_n, t_n} (s_n, \sigma_n)$$
in $\iota T$, and $\sigma_j = (R_{(s_{j-1}, a_j, s_j), \alpha}\circ\sigma_{s_{j-1}, \alpha}(t_j))_{\alpha \in I}$.
\end{lemma}

\begin{proof}
Let us prove it by induction on the length $n$. The case $n = 0$ is obvious. So now assume given a run 
$$s_0 \xrightarrow{a_1, t_1} \ldots \xrightarrow{a_n, t_n} s_n \xrightarrow{a_{n+1}, t_{n+1}} s_{n+1}$$
of $T$. Then $\sigma_n = (R_{(s_{n-1}, a_n, s_n), \alpha}\circ\sigma_{s_{n-1}, \alpha}(t_n))_{\alpha \in I} = (R_{(s_{n-1}, a_n, s_n), \alpha}\circ x_\alpha^{n-1}(t_n))_{\alpha \in I}$ following the notations above. $x_\alpha^n$ is then the solution of $\dot{x}(u) = F_{s_n, \alpha}(u, x(u))$ with the initial condition $x(0) = \sigma_{\alpha, n}$, where $\sigma_n = (\sigma_{\alpha, n})_{\alpha \in I}$. Furthermore, by definition of the guards and the invariants of $\iota T$:
\begin{itemize}
	\item $(x_\alpha^n(t_{n+1}))_{\alpha \in I} \in G_{(s_n, a_{n+1}, s_{n+1})}$,
	\item $(x_\alpha^n(t))_{\alpha \in I} \in I_{s_n}$ for every $t \leq t_{n+1}$.
\end{itemize}
So defining $\sigma_{n+1} = (R_{(s_n, a_{n+1}, s_{n+1}), \alpha})_{\alpha\in I}$, then $\iota T$ moves from $(s_n, \sigma_n)$ to $(s_{n+1}, \sigma_{n+1})$ by doing $a_{n+1}$ with time $t_{n+1}$. Furthermore, it is the only one since hybrid systems are deterministic when the action and the time are given.
\end{proof}

Given a state $s \in S$, $\eta_T$ is the run obtained from previous lemma, applied on the unique run to $s$ in $T$.

\begin{lemma}
$\eta_T$ is a $0$-bounded morphism of transition system with observation.
\end{lemma}

\begin{proof}
First, let us prove it is a morphism between the underlying transition systems. Given a transition $(s, (a, t), s') \in \Delta$, just apply the previous lemma on the unique run to $s$ and the unique run to $s'$ (which necessarily ends with the given transition). The unicity gives us that the second run $\eta_T(s')$ is an extension by $(a,t)$ of $\eta_T(s)$, which means there a transition of $H\iota T$ of the form $(\eta_T(s), (a, t), \eta_T(s'))$.

Now, given $s \in S$, 
$$d(\omega(s), \omega'(\eta_T(s))) = d(\omega(s), \omega(s)) = 0.$$
\end{proof}

Now, we want to construct the inverse $\map{\rho_T}{H\circ\iota T}$ of $\eta_T$. This morphism maps a run
$$(s_0, \sigma_0) \xrightarrow{a_1, t_1} \ldots \xrightarrow{a_n, t_n} (s_n, \sigma_n)$$
of $H\circ\iota T$ to $s_n$. 

\begin{lemma}
$\rho_T$ is a $0$-bounded morphism, which is the inverse of $\eta_T$.
\end{lemma}

\begin{proof}
Given a transition 
$$((s_0,\sigma_0) \xrightarrow{a_1, t_1} \ldots \xrightarrow{a_n,t_n} (s_n, \sigma_n), (a_{n+1}, t_{n+1}),~~~~~~~~~~~~~~~~~~~~$$
	$$~~~~~~~~~~~~~~~~~~~~(s_0,\sigma_0) \xrightarrow{a_1, t_1} \ldots \xrightarrow{a_n,t_n} (s_n, \sigma_n) \xrightarrow{a_{n+1}, t_{n+1}} (s_{n+1}, \sigma_{n+1}))$$
	of $H\circ\iota T$, we have in particular that $(s_n, a_{n+1}, s_{n+1})$ is an event of $\iota T$, which means that there is a unique $t$ such that $(s_n, (a_{n+1}, t), s_{n+1})$. $t$ is necessarily $t_{n+1}$ because, among the subsystems of $\iota T$, we have $\alpha = (1, 0, ((u,x)\in \mathbb{R}^2 \mapsto 1)_{s \in S}, (x \in \mathbb{R} \mapsto 0)_{e \in E})$. Then:
	\begin{itemize}
		\item The run in $H\circ\iota T$ gives us that $\sigma_{\alpha, n+1} = t_{n+1}$.
		\item The guard $G_{(s_n, a_{n+1}, s_{n+1}), \alpha}$ gives us that $\sigma_{\alpha, n+1} = t$.
	\end{itemize}

For the $0$-boundedness, given a run $$\pi = (s_0, \sigma_0) \xrightarrow{a_1, t_1} \ldots \xrightarrow{a_n, t_n} (s_n, \sigma_n)$$
of $H\circ\iota T$:
$$d(\omega'(\pi), \omega(\rho_T(\pi))) = d(\omega(s_n), \omega(s_n)) = 0.$$

For the fact that $\rho_T$ is the inverse of $\eta_T$:
\begin{itemize}
	\item $\rho_T\circ\eta_T = \text{id}$ is obvious.
	\item $\eta_T\circ\rho_T = \text{id}$ comes from the unicity of lemma \ref{lem:unicrun}
\end{itemize}
\end{proof}

\begin{lemma}
$\rho$ (and so $\eta$) is a natural transformation.
\end{lemma}

\begin{proof}
Given a morphism $f$ from $T$ to $T'$, $H\iota T$ maps a run 
$$(s_0, \sigma_0) \xrightarrow{a_1, t_1} \ldots \xrightarrow{a_n, t_n} (s_n, \sigma_n)$$
of $\iota T$ to a run of the form
$$(f(s_0), \sigma'_0) \xrightarrow{a_1, t_1} \ldots \xrightarrow{a_n, t_n} (f(s_n), \sigma'_n).$$
So both $f \circ \rho_T$ and $\rho_{T'}\circ H\iota f$ maps a 
$$(s_0, \sigma_0) \xrightarrow{a_1, t_1} \ldots \xrightarrow{a_n, t_n} (s_n, \sigma_n)$$
of $\iota T$ to $f(s_n)$.
\end{proof}

\begin{lemma}~
\begin{itemize}
	\item For every synchronization tree with observations $T$, $\epsilon_{\iota T} = \iota(\rho_T)$.
	\item For every hybrid system $T$, $H(\epsilon_T) = \rho_{HT}$.
\end{itemize}
\end{lemma}

\begin{proof}~
\begin{itemize}
	\item Fix $\epsilon_{\iota T} = (f_M, f_I)$ and $\iota(\rho_T) = (g_M, g_I)$. Then:
		\begin{itemize}
			\item Both $f_M$ and $g_M$ maps every run 
			$$(s_0, \sigma_0) \xrightarrow{a_1, t_1} \ldots \xrightarrow{a_n, t_n} (s_n, \sigma_n)$$
			of $\iota T$ to $s_n$.
			\item Both $f_I$ and $g_I$ maps a subsystem $(n, \sigma, (F_s)_{s\in S}, (R_e)_{e \in E})$ of $\iota T$ to the subsystem $$(n, \sigma, (F'_\pi)_{\pi \text{ run of } \iota T}, (R'_{e'})_{e' \text{ extension of runs of }\iota T})$$
			where:
			\begin{itemize}
				\item if $\pi$ is the run
				$$(s_0, \sigma_0) \xrightarrow{a_1, t_1} \ldots \xrightarrow{a_n, t_n} (s_n, \sigma_n)$$
				of $\iota T$, $F'_\pi = F_{s_n}$.
				\item if $e'$ is of the form 
				$$((s_0,\sigma_0) \xrightarrow{a_1, t_1} \ldots \xrightarrow{a_n,t_n} (s_n, \sigma_n), a_{n+1},~~~~~~~~~~~~~~~~~~~~$$
	$$~~~~~~~~~~~~~~~~~~~~(s_0,\sigma_0) \xrightarrow{a_1, t_1} \ldots \xrightarrow{a_n,t_n} (s_n, \sigma_n) \xrightarrow{a_{n+1}, t_{n+1}} (s_{n+1}, \sigma_{n+1}))$$
	$R'_{e'} = R_{(s_n, a_{n+1}, s_{n+1})}$.
			\end{itemize}
		\end{itemize}
	\item Given a run $\gamma$
	$$(\pi_0, \sigma_0) \xrightarrow{a_1, t_1} \ldots \xrightarrow{a_n, t_n} (\pi_n, \sigma_n)$$
	of $\iota H T$, there are $m_j$, $\sigma'_j$, such that $\pi_j$ is the run
	$$(m_0, \sigma'_0) \xrightarrow{a_1, t_1} \ldots \xrightarrow{a_j, t_j} (s_j, \sigma'_j)$$
	of $T$. Then $H(\epsilon_T)(\gamma)$ is the run
	$$(m_0, f_{T,X}(\sigma_0)) \xrightarrow{a_1, t_1} \ldots \xrightarrow{a_n, t_n} (m_n, f_{T,X}(\sigma_n))$$
	while $\rho_{HT}(\gamma) = \pi_n$. So to conclude, it is enough to prove that $(\sigma'_{i,j})_{i \in I} = \sigma'_j = f_{T,X}(\sigma_j) = f_{T,X}((\sigma_{\alpha, j})_{\alpha \in I'}) = (\sigma_{f_{T,I}(i), j})_{i \in I}$, by induction on $j$:
		\begin{itemize}
			\item \textbf{case $j=0$:} consequence of the fact that $\epsilon_T$ is a morphism, and so that $f_{T,X}$ preserves the initial valuation.
			\item \textbf{inductive case:} Assume that $\sigma'_j = f_{T,X}(\sigma_j)$. Then:
				\begin{itemize}
					\item $\sigma'_{i,j+1}$ is obtained as $R_{(m_{j}, a_{j+1}, s_{j+1}), i}(x(t_{j+1}))$ where $x$ is the solution on $[0, t_{j+1}]$ of $\dot{x}(u) = F_{m_{j}, i}(u, x(u))$, with the initial $x(0) = \sigma'_{i,j}$.
					\item for $\alpha = (n, \sigma, (F_\pi)_{\pi \text{ run of } T}, (R_e)_{e \text{ extension of runs of} T})$, $\sigma_{\alpha, j+1}$ is obtained as $R_{(\pi_{j}, a_{j+1}, \pi_{j+1})}(y(t_{j+1}))$, where $y$ is the solution on $[0, t_{j+1}]$ of $\dot{y}(u) = F_{\pi_{j}}(u, y(u))$, with the initial condition $y(0) = \sigma_{\alpha, j}$.
				\end{itemize}
			When $\alpha = f{T,I}(i)$, we have:
				\begin{itemize}
					\item $\sigma_{\alpha, j} = \sigma'_{i,j}$ by induction hypothesis,
					\item $R_{(\pi_{j}, a_{j+1}, \pi_{j+1})} = R_{(m_{j}, a_{j+1}, s_{j+1}), i}$ by definition of $f_{T,I}$,
					\item $F_{\pi_{j}} = F_{m_{j}, i}$ bu definition of $f_{T,i}$.
				\end{itemize}
			So $x$ and $y$ are solutions of the differential equations with the same initial conditions, so by Picard-Lindel\"of's theorem, $x = y$, and $\sigma'_{i,j+1} = \sigma_{f_{T,I}(i), j+1}$.
		\end{itemize}
\end{itemize}
\end{proof}

\end{document}